\documentclass[10pt,oneside,a4paper,reqno]{amsart}  %%(fold)
\usepackage{mathrsfs}

\usepackage{amssymb,%bbm,
enumerate}
\usepackage{bbm}
\usepackage{geometry}                % See geometry.pdf to learn the layout options. There are lots.
\usepackage[T1]{fontenc}
\usepackage[english]{babel}

\usepackage{color}
\usepackage[colorlinks=true, linkcolor=magenta, citecolor=cyan, urlcolor=blue]{hyperref}%link alle eq e ref
\renewcommand{\phi}{\varphi}
\renewcommand{\ker}{\Ker}

\newcommand{\mc}[1]{\mathcal{#1}}

\newcommand{\mb}[1]{\mathbb{#1}}
\newcommand{\tint}{{\textstyle\int}}

\def\be{\begin{equation}}
\def\ee{\end{equation}}
\def\bea{\begin{eqnarray}}
\def\eea{\end{eqnarray}}

\def\nn{\nonumber}
\def\T{\mathbb{T}}

\def\Z{\mathbb{Z}}

\DeclareMathOperator{\Mat}{Mat}

\DeclareMathOperator{\Ker}{Ker}
\DeclareMathOperator{\Tr}{Tr}

\DeclareMathOperator{\dd}{dd}
\DeclareMathOperator{\Span}{span}
\DeclareMathOperator{\ord}{ord}
\DeclareMathOperator{\Var}{Var}

\newcommand{\E}[1]{\mathbb{E}\left[#1\right]}

\def\ie{\textit{i.e. }}

\def\a{\alpha}
\def\e{\varepsilon}
\def\d{\delta}
\def\g{\gamma}

\def\b{\beta}
\def\D{\Delta}
\def\l{\lambda}
\def\r{\rho}
\def\s{\sigma}

\def\R{\mathbb{R}}
\def\C{\mathbb{C}}
\def\E{\mathbb{E}}
\newcommand{\ass}[1]{\stackrel{#1}{\longleftrightarrow}}

\theoremstyle{plain}
\newtheorem{theorem}{Theorem}[section]
\newtheorem{lemma}[theorem]{Lemma}
\newtheorem{proposition}[theorem]{Proposition}
\newtheorem{corollary}[theorem]{Corollary}

\theoremstyle{definition}
\newtheorem{definition}[theorem]{Definition}

\theoremstyle{remark}
\newtheorem{remark}[theorem]{Remark}

\numberwithin{equation}{section}

\definecolor{light}{gray}{.9}

\author{Giuseppe Genovese}
\address{Giuseppe Genovese: Institut f\"ur Mathematik, Universit\"at Z\"urich,
CH-8057 Z\"urich, Switzerland.}
\email{giuseppe.genovese@math.uzh.ch}

\author{Renato Luc\`a}
\address{Renato Luc\`a: Instituto de Ciencias Matem\'{a}ticas, CSIC, 28049 Madrid, Spain.}
\email{renato.luca@icmat.es}

\author{Daniele Valeri}
\address{Daniele Valeri: Yau Mathematical Sciences Center, Tsinghua University, 100084 Beijing, China.}
\email{daniele@math.tsinghua.edu.cn}

\title[Gibbs measures and DNLS equation]
{Gibbs measures associated to the integrals of motion of the periodic
derivative nonlinear Schr\"odinger equation}

\subjclass[2000]{35Q30, 35BXX, 37K05, 37L50, 35Q55, 37K10, 37K30, 17B69, 17B80}

\keywords{Gibbs measures, DNLS, integrable systems}

\makeindex

\begin{document}

\begin{abstract}
We study the one dimensional periodic derivative nonlinear Schr\"odinger (DNLS) equation. 
This is known to be a completely integrable system, in the sense that there is an infinite sequence
of formal integrals of motion $\tint h_k$, $k\in \mathbb{Z}_{+}$.
In each $\tint h_{2k}$ the term with the highest regularity involves the Sobolev norm $\dot H^{k}(\T)$
of the solution of the DNLS equation.
We show that
a functional measure on $L^2(\T)$, absolutely continuous w.r.t. the Gaussian measure
with covariance $(\mathbb{I}+(-\D)^{k})^{-1}$,
is associated to
each integral of motion $\tint h_{2k}$, $k\geq1$.

%\vspace{0.5cm}
%
%\ni {\bf MSC:} 35Q55, 37K10, 82B05.
\end{abstract}

\maketitle

%\tableofcontents

%%%%%%%%%%%%%%%%%%%%%%%%%%%%%%%%%%%%%%%%%%%%%%%%%%%%%%%%%%%%%
\section{Introduction}\label{sec:intro}

In this paper we consider the derivative nonlinear Schr\"odinger equation (DNLS) in the space periodic setting: 
\begin{equation}\label{eq:dnls}
\left\{
\begin{array}{lcl}
i\psi_t & = & \psi''+i\beta\left(\psi|\psi|^2\right)'  \\
\psi(x,0) & = &  \psi_{0}(x) \,, 
\end{array}
\right.
\end{equation}
where $\psi(x,t) : \mathbb{T} \times \mathbb{R} \rightarrow \mathbb{C}$,
$\psi_{0}(x) : \mathbb{T} \rightarrow \mathbb{C}$,
$\psi'(x,t)$ denotes the derivative with respect to $x$, and 
$\b\in\R$ is a real parameter.

The DNLS is a dispersive nonlinear equation coming from magnetohydrodinamics. It describes the motion along the longitudinal direction of a circularly polarised wave, generated in a low density plasma by an external magnetic field \cite{rog, moj} (see also \cite{sulemsulem}). It is known to be an integrable system \cite{KN78} (see also \cite{DSK13})
in the sense that there is an infinite sequence of linearly independent quantities (integrals of motion) which 
are conserved by the flow of \eqref{eq:dnls} for sufficiently regular solutions. The %conservation laws
integrals of motion are functionals defined on Sobolev spaces of increasing regularity.

The aim of this paper is to construct an infinite sequence of functional Gibbs measures associated to the integrals of motion.
These measures turn out to be absolutely continuous with respect to the standard Gaussian measures 
with covariance $(\mathbb{I}+(-\D)^{k})^{-1}$, thus different measures are disjointly supported (see Appendix A).

The program of statistical mechanics of PDEs begins with the seminal paper by Lebowitz, Rose and Speer
\cite{Leb}. The authors study the periodic one dimensional NLS equations and introduce the statistical
ensembles naturally associated to the Hamiltonian, as in a classical field theory.
Successively, Bourgain completed this study: in \cite{B94} by proving the invariance of the Gibbs measure
for the cubic periodic case and in \cite{Bou2} extending the results to the real line. Similar achievements have 
been obtained with different methods in \cite{MV2} for cubic NLS, in \cite{MV1} for the wave equation, and 
in two dimensions in \cite{Bou4} for defocusing cubic NLS equation, in \cite{BS} for the focusing case, in 
three dimensions for the Gross-Pitaevskii equation in \cite{BGP}.

For integrable PDEs one can carry out the same study by profiting from an infinite number of higher 
Hamiltonian functionals. This was originally noted by Zhidkov \cite{Zh01}, who analysed the
Korteweg-de Vries (KdV) and cubic nonlinear Schr\"odinger (NLS) equation on $\T$.
The main idea, already contained in \cite{Leb}, is to restrict the measure associated 
to a given 
integral of motion to the set of solutions with fixed values for all the other 
integrals of motion involving less regularity (in a sense that will be clarified below). The invariance of such a 
set of measures gives interesting informations on the long time behavior of the regular solutions, for instance through 
the Poincar\'e Recurrence Theorem (see \cite{Zh01, BTT14}).
In the last years this approach has been adopted in a series of papers by 
Tzvetkov, Visciglia and Deng \cite{TzPTRF,TV13a,TV13b,TV14,D14,DTV14} for the Benjamin Ono equation 
on $\T$. In this case %because of the nature of the equation, 
a more careful construction of the measure is 
required compared to KDV and NLS. 
%A major difficulty compared to KdV and NLS is that the non linearity has a non-trivial one-derivative 
%loss. 
We find similar difficulties in studying the DNLS equation.

Despite the extensive investigation in the past decades on integrable PDEs, a limited attention has been given
to the integrability properties of the DNLS equation. An infinite sequence of integrals of motion for this equation 
has been found in \cite{KN78} using the inverse scattering method. More recently, another proof of the 
integrability of the DNLS equation has been achieved using the so-called the Lenard-Magri scheme \cite{Mag78} within the 
framework of (non local) Poisson Vertex Algebras \cite{DSK13}.

The first few integrals of motion of the DNLS equation are:
\begin{align}
\label{Eq:Legge0}
\tint h_0 & 
=
\frac12\|\psi\|_{L^2}^2\,,
\\
\nonumber%\label{Eq:Legge1/2}
\tint h_1 &
=
\frac i2\int\psi\bar\psi'
+\frac\beta4\|\psi\|_{L^4}^4\,,
\\
\label{Eq:Legge1}
\tint h_2 &
=
\frac12 \|\psi\|_{\dot{H}^1}^2
+\frac{3i}2\beta\int\psi^2\bar\psi\bar\psi'
%\|\psi^2\|_{\dot{H}^{\frac12}}^2
+\frac{\beta^2}4\|\psi\|_{L^6}^6\,,
\\
\nonumber%\label{eq:Legge3/2}
\tint h_3 &
=\frac i2\int\psi'\bar\psi''
+\frac\beta4\int\left((\psi')^2\bar\psi^2
+8\psi\bar\psi\psi'\bar\psi'+\psi^2(\bar\psi')^2
\right)
+\frac{5 i}{4}\beta^2\int\psi^3\bar\psi^2\bar\psi'
+\frac 5{16}\beta^3\|\psi\|_{L^8}^8\,,
\\
\nonumber%\label{eq:Legge2}
\tint h_4 
&
= 
\frac12\|\psi\|_{\dot{H}^2}^2
+\frac{5i}4\beta\int\left(
\psi\bar\psi\psi'\bar\psi''-\psi\bar\psi\psi''\bar\psi'
\right)
\\
\nonumber
&
+
\frac{5}{4}\beta^2\int\left(
\psi\bar\psi^3(\psi')^2+5\psi^2\bar\psi^2\psi'\bar\psi'+\psi^3\bar\psi(\bar\psi')^2
\right)
+\frac{35i}{16}\beta^3\int\psi^4\bar\psi^3\bar\psi'
+\frac{7}{16}\beta^4\|\psi\|_{L^{10}}^{10}\,.
\end{align}
Here and further, we denote by $\int f=\frac{1}{2\pi}\int_\T f$.
Note that, while for $k$ even the term of highest regularity in the integral of motion $\int h_{k}$
is the $\dot H^{\frac k2}(\T)$ norm, for odd $k$ this term is not definite in sign. This prevents us to 
associate an invariant Gibbs measure to every integral of motion $\tint h_k$, $k\in\mb Z_+$.
The same does not occur for KdV, NLS or Benjamin-Ono equations. 

The DNLS equation has been shown to be locally well posed for initial data in $H^{s\geq 1/2}$ both for
periodic and non periodic settings (see \cite{Herr} and respectively \cite{Tk}). The global well-posedness 
has been proven for $H^{s\geq 1/2}(\R)$ in %\cite{I-team1} 
\cite{miao} and in $H^{s>1/2}(\T)$ in \cite{Win}. The global results hold for initial data with small $L^2(\T)$ norm.
For instance a standard procedure (see \cite{Herr}) allows to globalize the local $H^{1}(\mathbb{T})$ solutions
provided that $\| \psi_{0} \|_{L^{2}(\mathbb{T})} < \delta$ with $\delta$ small enough, by using the conservation law $\int h_{2}$ and the  
Gagliardo--Nirenberg inequality
\be\label{eq:GN}
\|u\|^3_{L^6(\mathbb{T})}\leq
\|u\|_{\dot{H}^{1}(\mathbb{T})}\|u\|_{L^2(\mathbb{T})}^{2}
+
\frac{1}{2 \pi} \|u\|^3_{L^2(\mathbb{T})}\,.
\ee
%Throughout the paper we use the notation $\|\psi\|_{L^2}\ll1/\sqrt{|\b|}$ to indicate that (\ref{eq:L2small}) holds true.
On the other hand this approach does not give the best possible value for $\delta$, which is still unknown. 
In particular in the case of DNLS on $\mathbb{R}$ the
sharp Gagliardo--Nirenberg inequality
$$%\begin{equation}
\|u\|^3_{L^6(\mathbb{R})}\leq
\frac{2}{\pi}
\|u\|_{\dot{H}^{1}(\mathbb{R})}\|u\|_{L^2(\mathbb{R})}^{2} \, ,
$$%\end{equation}
proved in \cite{Wei82},  gives the value $\delta = \sqrt{2\pi / |\beta|}$ for global well-posedness \cite{HO92, HO93}, which has been actually improved to $\delta=2\sqrt{\pi / |\beta|}$ 
in \cite{Wu13, Wu15}
by different techniques.
We point out that $\delta=2\sqrt{\pi / |\beta|}$ is also sufficient on $\mathbb{T}$ \cite{MO15}.
All this results are originally stated for $\beta = \pm 1$, the case of general $\beta$ can be easily deduced by
using the transformation $u(t,x) \to |\beta|^{-1/2} u(t, \frac{\beta}{|\beta|} x)$.

The lack of well-posedness at low regularity makes hard to construct an invariant measure associated to
the lowest order integrals of motion. 
%conservation laws.
For $\tint h_2$ the main issue is that there is no well-posedness in 
$\bigcap_{\e>0}H^{\frac12-\e}(\T)$ which is the support of the Gaussian measure with covariance 
$\mathbb{I}-\D$. A very delicate analysis is necessary to deal with this problem. In \cite{NOR-BS12} the authors constructed a functional measure in the Fourier-Lebesgue space $\mathcal{F}L^{s,r}(\T)$, 
$r\in(2,4)$, $s\in[1/2,1-r^{-1})$, for which there is a local existence theorem \cite{GH}. They are 
able to prove the invariance of this measure with respect to the DNLS flow (studying in fact the gauged 
DNLS equation). Then in \cite{NR-LSS11} the study is completed, by proving the absolute continuity of this measure with respect to the Gibbs measure constructed in \cite{TT10}, which would be a more natural candidate for the invariant measure associated to the energy functional $\tint h_2$. 
Similar results for the DNLS equation have been obtained, with different methods, in \cite{BTT14}.  
To the best of our knowledge, so far these are the sole known results for Gibbs measures associated to the DNLS equation.

%To the $k$-th conservation law one can associate a measure in $L^2(\T)$.
%The goal is to construct it as an absolutely continuous measure w.r.t. a Gaussian one defined as follows: for 
%each $k\geq1$ the Gaussian measure $d\g_k$ on $L^2(\T)$ with covariance operator
%$\D^{-k}$ is 
%determined by its finite dimensional marginal distributions:
%$$
%\g_k^N(A)=\prod_{n=0}^{N-1}\left[\frac{\left(1+n^{2k}\right)}{2\pi}\right]^{\frac12}\int_A du_0...du_{N-1} e^{-\frac12\sum_{n=0}^{N-1}(1+n^{2k})|u_n|^2},
%$$
%for any $A$, Borel subset of $\R^N$, $N\geq1$, and $u_0,...,u_{N-1}$ are the first $N$ Fourier coefficients of $u\in L^2(\T)$. We refer to \cite{Bog} for a detailed presentation.
%
%Let $R>0$, and let $\chi_R:\mb R\to[0,1]$ be a function in $C_0([-R,R])$, with $\chi=1$ in $[-\frac R2,\frac R2]$.
%%
%Fix $R_m>0$, for $m=0,\dots,k-1$.
%%
%For $u\in L^2\T$, we define the density 
%\begin{equation}\label{eq:G_k}
%G_{k,N}(u)
%=\Big(
%\prod_{m=0}^{k-1}
%\chi_{R_m}\left(\|u_N\|_{H^{m}}\right)
%\Big)
%e^{-\tint q_{k}(u_N)}
%\,,
%\end{equation}
%where $\tint q_k$ denotes the second sum in \eqref{eq:H_k}.
%The associated measure $d\rho_{k,N}$ is
%\begin{equation}\label{eq:rho_k,N}
%d\rho_{k,N}(u)
%=G_{k,N}(u)d\gamma_{k}(u)\,.
%\end{equation}
%
\subsection{Set up and Main Result}

The main goal of this paper is to construct Gaussian measures supported on Sobolev spaces with increasing 
regularity, associated to the integrals of motion of the DNLS. Let us introduce now the main objects we are 
going to deal with. 

As usual we denote by $H^k(\T)$, $k\in\mb Z_+$, the completion of $C^{\infty}(\T)$ with respect to the norm induced by the scalar product
$$
(u,v)_{H^k}:=\sum_{n\in\Z} (1+|n|)^{2k} \bar u_n v_n \,.
$$
where $u_n$ are the Fourier coefficients of $u$. For every $k\in\mb Z_+$, $H^k(\T)$ is a separable
Hilbert space, and we note that $H^{0}(\T)=L^2(\T)$. A function in $H^k(\T)$ is represented as a 
sequence $\{u_n\}_{n\in\mb Z_+}$ such that $\sum_{|n|\leq N}(1+ |n|)^{2k}|u_n|^2$ is finite uniformly in $N\in\mathbb{Z}_+$.
%\bea
%\|u\|_{H^k}&=&\sum_n|u_n|^2+\sum_n n^{2k}|u_n|^2\\
%(u,v)_k&=&\sum_n u_n^*v_n+\sum_n n^{2k} u^*_n v_n,
%\eea

We also use the homogeneous Sobolev spaces $\dot H^k(\T)$, defined as the completion of $C^{\infty}(\T)$ with respect to the norm induced by the homogenous scalar product
$$
(u,v)_{\dot H^k}:=\sum_{n\in\Z} n^{2k} \bar u_n v_n \,.
$$

%We remind some well known calculus facts that will be repeatedly used in the paper. Let $u : \T \rightarrow \mathbb{C}$ be a periodic function. We have the Sobolev embedding
%\begin{equation}\label{Eq:S0bEmbedd}
%\|f\|_{\infty}
%\lesssim \|f\|_{H^1}\,,
%\end{equation}
%and the Gagliardo--Nirenberg inequality
%\begin{equation}\label{eq:GN_torus}
%\|u\|^3_{L^6}\leq
%K_{0} \|u\|_{\dot{H}^{1}}\|u\|_{L^2}^{2}
%+
%K_{1} \|u\|^3_{L^2}\,.
%\end{equation}

%%%%%%%

Now we consider the Hilbert space $L^2(\T)$.
For any $k\in\mb Z_+$, let us denote by $\mathbb{I}+(-\D)^k$ the closure in $L^2(\T)$ of the operator 
$1+\left(-\frac{d^{2}}{d x^{2}}\right)^k$ acting on $C^{\infty}(\T)$. As it is well known this is a positive, self adjoint 
operator with a trivial kernel. Therefore its inverse $(\mathbb{I}-\D^k)^{-1}$ is bounded and moreover it 
can be shown that it is of trace class. 

In virtue of this last property we can construct a Gaussian measure on $L^2(\T)$ as follows.
We denote by $e_n=e^{inx}$ the eigenvectors of $\mathbb{I}+(-\D)^k$:
$$
(\mathbb{I}+(-\D)^k)e_n= (1+n^{2k})e_n.
$$
Since $\mathbb{I}+(-\D)^k$ is self-adjoint the set of its eigenvectors spans the space $L^2(\T)$, and so each function $u(x)\in L^2(\T)$ can be written as
$$
u(x)=\sum_{n\in\mathbb{Z}} u_ne_n,
$$
that is nothing but Fourier series. We consider at first finite dimensional truncations, looking only at the components of the expansion for $|n|\leq N$. %Therefore we denote by $(\cdot, \cdot)_N$ the scalar product in $\R^{2N+1}$ and by $\D_N:=-\diag(n^2)_{|n|\leq N}$. 
We define
$$%\be\label{eq:misGauss}
\g_k^N(A):=\frac{\prod_{|n|\leq N} \sqrt{1+n^{2k}}}{(2\pi)^{2N+1}}\int_A du_{-N}d\bar u_{-N}...du_Nd\bar u_N e^{-\frac12\sum_{|n|\leq N}(1+n^{2k})|u_n|^2}
$$%\ee
to be the complex Gaussian measure of a set $A\subseteq \C^{2N+1}$. This measure can be extended in infinite dimensions following a standard method \cite{Sko, Zh01}.
For any Borel subset $B\subset\C^{2N+1}$ we introduce the corresponding cylindrical set in $L^2(\T)$ as
$$
M_N(B)=\{u\in L^2(\T)\mid [%(e_{-N},u)_{H^k},\dots,(e_N,u)_{H^k},
u_{-N},\bar u_{-N},\dots,u_N, \bar u_{N}]\in B\}\,.
$$
Since $\mathbb{I}+(-\D)^k$ is of trace class, we can extend the Gaussian measure $\g_k^N$ to $L^2(\T)$ functions by setting $\g_k(M_N):=\g_k^N(M_N)$ and then using Kolmogorov reconstruction theorem. It can be verified that this defines a countably additive measure on $L^2(\T)$. We refer to \cite{Bog, Sko} for a more detailed presentation (see also Appendix A for some properties that will be used in the paper). We denote by $L^p_{\g_k}$ the Banach space of functionals $F\,:\,L^2(\T)\to\mathbb{C}$ such that
$$
\int d\g_k(d\psi) |F(\psi)|^p<\infty.
$$
For the ease of notation we simply denote as $\E[\cdot]$ (instead of $\E_{\g_k}[\cdot]$) the expectation value w.r.t. the measure $\g_k$. Anyway the particular $\g_k$ considered will be always clear from the context. 

\medskip

For $N\geq1$, we set $E_{N}=\Span_{\mb C}\{e^{inx}\mid |n|\leq N\}$,
and we denote by $P_{N} : L^{2}(\T) \to E_N$ the projection map onto the
space $E_N$. Namely, for $u=\sum_{n\in\mb Z}u_ne^{inx} \in L^{2}(\T)$,
we have
\be\label{eq:proiettore}
P_Nu:=\sum_{|n|\leq N}u_ne^{inx}\,.
\ee
When there is no confusion, we simply denote 
\be \label{eq:u_N}
u_N:=P_Nu\,.
\ee

For $\psi\in L^2(\T)$, we show in Section \ref{sec:1} that 
$$
\tint h_{2k}[\psi]=\frac12 \|\psi\|^2_{\dot H^{k}}+\tint q_k[\psi]\,,
\qquad
k\in\mb Z_+\,,
$$
where $\tint q_k$ is a sum of terms of the form
$$
\tint\bar\psi^{(\a_1)}\dots\bar\psi^{(\a_l)}\psi^{(\b_1)}\dots\psi^{(\b_l)},
$$
with $l\leq 2k+2$, $\a_i,\b_i\in\mb Z_+$ and $\sum_{i\leq l}\a_i+\b_i\leq 2k-1$.

Let now $R>0$, and let $\chi_R:\mb R\to[0,1]$ be a smooth function such that
$\chi = 0 \in \R \setminus[-R,+R]$ and $\chi=1$ in $[-R/2, +R/2]$.
For $k\geq2$, let us fix $R_m>0$, for $m=0,\dots,k-1$.
Thus we can define the density 
\begin{equation}\label{eq:G_k}
G_{k,N}(\psi)
=\Big(
\prod_{m=0}^{k-1}
\chi_{R_m}\left(\tint h_{2m}(\psi_{N})\right)
\Big)
e^{-\tint q_{k}(\psi_N)}
\,.
\end{equation}
The associated measure $d\rho_{k,N}$ is
$$%\begin{equation}\label{eq:rho_k,N}
\rho_{k,N}(d\psi)
=G_{k,N}(\psi)\gamma_{k}(d\psi)\,.
$$%\end{equation}

The main result of the paper is the following:
\begin{theorem}\label{thm:k_even}
Let $k\geq2$ and $R_0 \leq \sqrt{\frac{2}{9 | \beta |}}$ sufficiently small. 
The sequence $G_{k,N}(\psi)$ defined by equation \eqref{eq:G_k}
converges in measure, as $N\to\infty$, w.r.t. the measure
$\gamma_{k}$.
Denote by $G_k(\psi)$ its limit.
Then, there exists $p_0(R_0,\dots,R_{k-1},k,|\beta|)>1$ such that, for all
$1\leq p<p_0$, $G_k(\psi)\in L^p(\g_k)$ and
$G_{k,N}(\psi)$ converges to $G_k(\psi)$ in $L^p(\g_k)$.
\end{theorem}
\begin{remark}
The best range one should expect to obtain for $R_0$ is the same
of the global-wellpsedness problem, which at the moment is $R_0 < 2\sqrt{\frac{\pi}{|\beta|}}$.
The fact that we only get $R_0 \leq \sqrt{\frac{2}{9 | \beta |}}$ is, as we have observed 
above, a limitation of the Gagliardo--Nirenberg inequality approach.
Therefore one could presumably obtain the widest range $\left[ 0, 2\sqrt{\frac{\pi}{|\beta|}} \right)$ by using different techniques. 
We point out that this
could improve the absolute constant, while the behavior $\simeq 1 / \sqrt{ |\b| } $ seems to be a peculiar feature of the equation.
\end{remark}
As a consequence of Theorem \ref{thm:k_even}, we obtain that the measures
$\rho_{k,N}$ weakly converge, as $N\to\infty$, to the Gibbs measures $\rho_k$ on $L^2(\T)$:
$$%\begin{equation}\label{eq:rho_k}
\rho_k(d\psi)=G_k(\psi)\gamma_{k}(d\psi)\,.
$$%\end{equation}
Since each $G_k$ is supported on a set of positive measure w.r.t. $\gamma_{k}$, for every $k\geq2$, $\r_k$ is 
non trivial and absolutely continuous w.r.t. to $\gamma_{k}$. We choose the class observables associated to each of these Gibbs measure to be the functionals in $L^{\infty}_{\g_k}$.
%%%

\subsection{Strategy of the Proof}

The first part of our proof relies on an accurate inspection of the algebraic structures of the integrals of motion
of the DNLS equation.
This has been done in Section \ref{sec:1}. We use the Lenard-Magri scheme of integrability
for non local Poisson vertex algebras to find out the following general structure of the
integrals of motion:
$$
\int h_{2k}=\frac12\|\psi\|^2_{\dot H^k}+\frac i2\b(2k+1)\int \bar\psi^{(k)}\psi^{(k-1)}\bar\psi \psi+\mbox{ a remainder}\,,
\qquad
k\in\mb Z_+\,,
$$
where we consider as remainder all the terms that we can estimate with a certain power of the $H^{k-1}$ norm. Note that this quantity is finite in the support of the Gaussian measure $\g_k$.

In Section \ref{sec:stab} we show, under the $L^2$ smallness assumption,
that the Sobolev norm $H^{k}$ of the solutions of the DNLS equation \eqref{eq:dnls} stays bounded by a constant depending on the values of $\tint h_{2m}$, $m=1,\dots,k$, integrals of motion. 
Therefore, when we introduce the cut off functions $\chi$ in \eqref{eq:G_k}, we know that the $H^{s}$ norms, $s\leq k-1$, are bounded a.s. in the support of the Gibbs measure $\r_{k,N}$ uniformly in $N$. 
This allows us to prove in Section \ref{sec:CON} that all the remainder terms converge point-wise in the support of $\r_{k,N}$ as $N\to\infty$, thus also in measure w.r.t. $\g_k$. 

The terms $\int \bar\psi^{(k)}\psi^{(k-1)}\bar\psi \psi$ are estimated by the $H^k(\T)$ norm, which is not finite in the support of $\g_k$. Therefore they need to be treated separately. This is done by using a method outlined by Bourgain in \cite{Bou4} (see also \cite{BS}), which is reminiscent of the works in quantum field theory in the '70 \cite{guerra,Simon}. Successively this approach has been exploited by Tzvetkov and collaborators in \cite{TT10} for DNLS equation and in \cite{TzPTRF, TV13a} for the Benjamin Ono equation.

In Section \ref{sec:CON} we prove the convergence in $L^2(\g_k)$ of these terms as $N\to\infty$, employing essentially the Wick theorem. $L^2(\g_k)$ convergence yields $L^{p}(\g_k)$ ($p\in[1,\infty)$) convergence by a standard hyper-contractivity argument. This is enough to prove convergence in measure of the density. In Section \ref{sec:fin} we ultimate our strategy showing $L^p(\g_k)$ boundedness of the density $G_k$ for $p\in[1,\infty)$, provided that $\tint h_0$ is sufficiently small. Here we follow the nice ideas of \cite{TzPTRF}, making use of some helpful properties of the measures $\g_k$ reviewed in Appendix \ref{app-Gauss}. 

From the $L^p(\g_k)$ boundedness the convergence in $L^p(\g_k)$ (and so the weak convergence) of the density easily follows.

In the whole paper (except for Section \ref{sec:CON}) we are not concerned about 
the dynamics. However the measures that we construct are naturally expected to be invariant under the flow of DNLS. To prove this result, a careful analysis is required (as for instance in the case of the Benjamin-Ono equation \cite{TV13b, TV14, DTV14}) 
which we leave to a forthcoming work. 

\medskip

Throughout the paper we write $X \lesssim Y$ to denote that $X \leq C Y$ for some positive constant $C$ independent on $X,Y$.

\subsection*{Acknowledgments}
This work begun during the visit of the first and second author
to SISSA in Trieste. Then this research was supported through the programme
``Research in Pairs'' by the Mathematisches Forschungsinstitut Oberwolfach
in 2014. We thank these institutions for the kind hospitality. 
G.G. is supported by the Swiss National Science Foundation. R. L. is supported by the ERC grant 277778 and MINECO grant SEV-2011-0087 (Spain) and 
partially supported by the Italian Project FIRB 2012 ``Dispersive
dynamics: Fourier Analysis and Variational Methods''.
D.V. is partially supported by an NSFC 
``Research Fund for International Young Scientists'' grant.
We are grateful to B. Schlein for constant encouragement and many valuable comments.
%

%%%%%%%%%%%%%%%%%%%%%%%%%%%%%%%%%%%%%%%%%%%%%%%%%%%%%%%%%%%%%%%%
\section{Structure of the integrals of motion of the DNLS equation}
\label{sec:1}
In this section we recall briefly the theory of Poisson vertex algebras aimed at the
study of the integrability properties of bi-Hamiltonian equations using the so-called Lenard-Magri
scheme (see \cite{Mag78,BDSK09,DSK13}).
We use this formalism to describe explicitly 
the structure of the integrals of motion of the DNLS equation which will
be used throughout the paper.

\subsection{Algebras of differential polynomials}\label{sec:1.1}
Let $\mc V$ be the algebra of differential polynomials in $\ell$ variables:
$\mc V=\mb C [u_i^{(n)}\, |\, i \in I,n \in \mb Z_+]$, where $I=\{1,\dots,\ell\}$.
(In fact, most of the results hold in the generality of algebras 
of differential functions, as defined in \cite{DSK13}.)
It is a differential algebra with derivation defined by
$\partial(u_i^{(n)})=u_i^{(n+1)}$.
We also let $\mc K$ be the field of fractions of $\mc V$
(it is still a differential algebra).

For $P\in\mc V^\ell$ we have the associated \emph{evolutionary vector field}
$$
X_P=\sum_{i\in I,n\in\mb Z_+}(\partial^nP_i)\frac{\partial}{\partial u_i^{(n)}}
\,.
$$
This makes $\mc V^\ell$ into a Lie algebra, with Lie bracket
$[X_P,X_Q]=X_{[P,Q]}$, given by
$$
[P,Q]=X_P(Q)-X_Q(P)
=D_Q(\partial)P-D_P(\partial)Q
\,,
$$
where $D_P(\partial)$ and $D_Q(\partial)$ denote the Frechet derivatives of $P,Q\in\mc V^\ell$
(we refer to \cite{BDSK09} for the definition of Frechet derivative).

For $f\in\mc V$ its \emph{variational derivative} is
$\frac{\delta f}{\delta u}=\left(\frac{\delta f}{\delta u_i}\right)_{i\in I}\in\mc V^{\oplus\ell}$,
where
\begin{equation}\label{20141031:eq2}
\frac{\delta f}{\delta u_i}
=\sum_{n\in\mb Z_+}(-\partial)^n\frac{\partial f}{\partial u_i^{(n)}}\,.
\end{equation}
Given an element $\xi\in\mc V^{\oplus\ell}$, the equation $\xi=\frac{\delta h}{\delta u}$ can be solved
for $h\in\mc V$ if and only if $D_\xi(\partial)$ is a self-adjoint operator:
$D_\xi(\partial)=D_\xi^*(\partial)$ (see \cite{BDSK09}).

For $f\in\mc V$, we denote by $\tint f=f+\partial\mc V$, where $\partial\mc V=\{\partial h\mid h\in\mc V\}$,
the image of $f$ in the quotient space $\mc V/\partial\mc V$, and we call it a \emph{local functional}.
Note that the integral symbol is motivated by the fact that $\mc V/\partial\mc V$
provides a universal space where integration by parts holds, namely
$$
\tint f\partial g=-\tint g\partial f\,,
\qquad
\text{for every }f,g\in\mc V
\,.
$$
It is possible to show that $\ker\frac{\delta}{\delta u}=\partial\mc V\oplus\mb C$.
Hence,
$\frac{\delta f}{\delta u}=\frac{\delta \tint f}{\delta u}=0$. %(see \cite{BDSK09}).
Recall also that we have a non-degenerate pairing 
$(\cdot\,|\,\cdot):\,\mc V^\ell\times\mc V^{\ell}\to\mc V/\partial\mc V$
given by $(P|\xi)=\tint P\cdot\xi$ (see \cite{BDSK09}).

Given $f\in \mc V\backslash\mb C$,
we say that it has \emph{differential order} $n$,
and we write $\ord(f)=n$,
if $\frac{\partial f}{\partial u_i^{(n)}}\neq0$
for some $i\in I$ and $\frac{\partial f}{\partial u_j^{(m)}}=0$ for all $j\in I$ and $m> n$.
We also set the differential order of elements in $\mb C$ equal to $-\infty$.
Let us denote by $\mc V_n$ the space of polynomials of differential order at most $n$.
This gives an increasing sequence of subalgebras
$\mb C=\mc V_{-\infty}\subset\mc V_0\subset\mc V_1\subset\dots\subset\mc V$
such that $\partial\mc V_n\subset\mc V_{n+1}$.

We extend the notion of differential order to elements in $P\in\mc V^\ell$ as follows:
$$
\ord(P)=\max\{\ord(P_1),\dots,\ord(P_\ell)\}
\,.
$$

We also define two gradings on $\mc V$ in the following way.
First, we let $\deg$ be the usual polynomial grading of $\mc V$ defined by
$$
\deg u_i^{(n)}=1
\,,
\qquad
\text{for every }i\in I,n\in\mb Z_+\,.
$$
On the other hand we define the differential grading on $\mc V$, which
we denote $\dd$, by
$$
\dd u_i^{(n)}=n\,,
\qquad
\text{for every }i\in I,n\in\mb Z_+\,.
$$
This means that, given a monomial ($i_1,\dots,i_k\in I\,,n_1,\dots,n_k\in\mb Z_+$)
$$
f=u_{i_1}^{(n_1)}u_{i_2}^{(n_2)}\dots u_{i_k}^{(n_k)}\in\mc V
\,,
$$
we have
$$
\deg(f)=k\,,
\qquad
\dd(f)=n_1+\dots +n_k\,.
$$
Note that, for a homogeneous polynomial $f\in\mc V$, we have
\begin{equation}\label{20141102:eq4}
\deg(\partial f)=\deg(f)\,,
\qquad
\dd(\partial f)=\dd(f)+1
\,.
\end{equation}

%%%
\subsection{Rational matrix pseudodifferential operators and the association relation}\label{sec:1.2}

Consider the skewfield $\mc K((\partial^{-1}))$ of pseudodifferential operators 
with coefficients in $\mc K$,
and the subalgebra $\mc V[\partial]$ of differential operators on $\mc V$.

The algebra $\mc V(\partial)$ of \emph{rational} pseudodifferential operators
consists of pseudodifferential operators $L(\partial)\in\mc V((\partial^{-1}))$
which admit a fractional decomposition
$L(\partial)=A(\partial)B(\partial)^{-1}$,
for some $A(\partial),B(\partial)\in\mc V[\partial]$, $B(\partial)\neq0$.
The algebra of \emph{rational matrix pseudodifferential operators}
is, by definition,  $\Mat_{\ell\times\ell}\mc V(\partial)$ \cite{CDSK13}.

A matrix differential operator $B(\partial)\in\Mat_{\ell\times\ell}\mc V[\partial]$
is called \emph{non-degenerate}
if it is invertible in $\Mat_{\ell\times\ell}\mc K((\partial^{-1}))$.
Any matrix $H(\partial)\in\Mat_{\ell\times\ell}\mc V(\partial)$
can be written as a ratio of two matrix differential operators:
$H(\partial)=A(\partial) B^{-1}(\partial)$,
with $A(\partial),B(\partial)\in\Mat_{\ell\times\ell}\mc V[\partial]$,
and $B(\partial)$ non-degenerate.

Given $H(\partial)\in\Mat_{\ell\times\ell}\mc V(\partial)$,
we say that $\xi\in\mc V^{\oplus l}$ and $P\in\mc V^\ell$
are $H$-\emph{associated},
and denote it by
\begin{equation}\label{20130112:eq2}
\xi \ass{H}P\,,
\end{equation}
if there exist
a fractional decomposition $H=AB^{-1}$ 
with $A,B\in\Mat_{\ell\times\ell}\mc V[\partial]$ 
and $B$ non-degenerate,
and an element $F\in\mc K^\ell$,
such that $\xi=BF,\,P=AF$ \cite{DSK13}.

%%%
\subsection{Non-local Poisson structures}\label{sec:1.3}

A \emph{non-local Poisson vertex algebra}
is a differential algebra $\mc V$ endowed with a $\lambda$-bracket
$\{\cdot\,_\lambda\,\cdot\}:\,\mc V\times\mc V\to\mc V((\lambda^{-1}))$,
where $\mc V((\lambda^{-1}))$ denotes the space of Laurent series in $\lambda^{-1}$ with coefficients
in $\mc V$,
satisfying sesquilinearity ($f,g\in\mc V$):
$$
\{\partial f_\lambda g\}=-\lambda \{f_\lambda g\}\,,
\qquad
\{f_\lambda \partial g\}=(\lambda+\partial)\{f_\lambda g\}
\,,
$$
the Leibniz rule ($f,g,h\in\mc V$):
$$
\{f_\lambda gh\}=\{f_\lambda g\}h+\{f_\lambda h\}g
\,,
$$
skewsymmetry ($f,g\in\mc V$):
$$
\{f_\lambda g\}=-\{g_{-\lambda-\partial} f\}\,,
$$
admissibility ($f,g,h\in\mc V$):
$$
\{f_\lambda\{g_\mu h\}\}\in\mc V[[\lambda^{-1},\mu^{-1},(\lambda+\mu)^{-1}]][\lambda,\mu]
\,,
$$
and Jacobi identity ($f,g,h\in\mc V$):
$$
\{f_{\lambda}\{g_{\mu}h\}\}-\{g_{\mu}\{f_{\lambda}h\}\}
=\{\{f_\lambda g\}_{\lambda+\mu}h\}
\,.
$$
We refer to \cite{DSK13} for the details on the notation.

To a matrix pseudodifferential operator
$H=\big(H_{ij}(\partial)\big)_{i,j\in I}\in\Mat_{\ell\times\ell}\mc V((\partial^{-1}))$
we associate a $\lambda$-bracket,
$\{\cdot\,_\lambda\,\cdot\}_H:\,\mc V\times\mc V\to\mc V((\lambda^{-1}))$,
given by the following \emph{Master Formula} (see \cite{DSK13}):
\begin{equation}\label{20110922:eq1}
\{f_\lambda g\}_H
=
\sum_{\substack{i,j\in I \\ m,n\in\mb Z_+}} 
\frac{\partial g}{\partial u_j^{(n)}}
(\lambda+\partial)^n
H_{ji}(\lambda+\partial)
(-\lambda-\partial)^m
\frac{\partial f}{\partial u_i^{(m)}}
\,\in\mc V((\lambda^{-1}))
\,.
\end{equation}
%
%In particular,
%%\begin{equation}\label{20130613:eq2}
%$$
%H_{ji}(\partial)
%=
%{\{{u_i}_\partial{u_j}\}_H}_\to
%\,.
%$$
%%\end{equation}
%(The arrow means that we move $\partial$ to the right.)
%
For arbitrary $H$, it is proved in \cite{BDSK09} and \cite{DSK13}, that the
$\lambda$-bracket \eqref{20110922:eq1}
satisfies sesquilinearity and the Leibniz rule.
Furthermore, it has been shown that skewadjointness of $H$ is equivalent to the
skewsymmetry condition,
and that, if $H$ is a rational matrix pseudodifferential operator, then the
admissibility condition holds.
\begin{definition}\label{20140117:def}
A \emph{non-local Poisson structure} on $\mc V$
is  a skewadjoint rational matrix pseudodifferential operator $H$ with coefficients in $\mc V$
such that the corresponding $\lambda$-bracket \eqref{20110922:eq1}
satisfies Jacobi identity, namely, $\mc V$ endowed with the $\lambda$-bracket \eqref{20110922:eq1}
is a non-local Poisson vertex algebra.
\end{definition}

Two non-local Poisson structures $H,K\in\Mat_{\ell\times\ell}\mc V(\partial)$ on $\mc V$
are said to be \emph{compatible} if any of their linear combination (or, equivalently, their sum)
is a non-local Poisson structure.
In this case we say that $(H,K)$ form a \emph{bi-Poisson structure} on $\mc V$.
%

%%%
\subsection{Hamiltonian equations and integrability}\label{sec:1.4}

Let $H\in\Mat_{\ell\times\ell}\mc V(\partial)$ be a non-local Poisson structure.
An evolution equation on the variables $u=\big(u_i\big)_{i\in I}$,
\begin{equation}\label{20120124:eq5}
\frac{du}{dt}
=P\in\mc V^\ell\,,
\end{equation}
is called \emph{Hamiltonian} with respect to the non-local Poisson structure $H$
and the Hamiltonian functional $\tint h\in\mc V/\partial\mc V$
if (see Section \ref{sec:1.2})
$$
\frac{\delta h}{\delta u}\ass{H}P
\,.
$$

Equation \eqref{20120124:eq5} is called \emph{bi-Hamiltonian}
if there are two compatible non-local Poisson structures $H$ and $K$,
and two local functionals $\tint h_0,\tint h_1\in\mc V/\partial\mc V$,
such that
\begin{equation}\label{20140117:eq1}
\frac{\delta h_0}{\delta u}\ass{H}P
\qquad\text{ and }\qquad
\frac{\delta h_1}{\delta u}\ass{K}P
\,.
\end{equation}
An \emph{integral of motion} for the Hamiltonian equation \eqref{20120124:eq5}
is a local functional $\tint f\in\mc V/\partial\mc V$
which is constant in time, namely, such that $(P|\frac{\delta f}{\delta u})=0$.
The usual requirement for \emph{integrability}
is to have
sequences $\{\tint h_n\}_{n\in\mb Z_+}\subset\mc V/\partial\mc V$ 
and $\{P_n\}_{n\in\mb Z_+}\subset\mc V^\ell$,
starting with $\tint h_0=\tint h$ and $P_0=P$,
such that
\begin{enumerate}[(C1)]
\item
$\frac{\delta h_n}{\delta u}\ass{H}P_n$ for every $n\in\mb Z_+$,
\item
$[P_m,P_n]=0$ for all $m,n\in\mb Z_+$,
\item
$(P_m\,|\,\frac{\delta h_n}{\delta u})=0$ for all $m,n\in\mb Z_+$.
\item
The elements $P_n$ span an infinite dimensional subspace of $\mc V^\ell$.
\end{enumerate}
In this case, we have an \emph{integrable hierarchy} of Hamiltonian equations
$$
\frac{du}{dt_n} = P_n\,,\,\,n\in\mb Z_+\,.
$$
Elements $\tint h_n$'s are called \emph{higher Hamiltonians},
the $P_n$'s are called \emph{higher symmetries},
and the condition $(P_m\,|\,\frac{\delta h_n}{\delta u})=0$
says that $\tint h_m$ and $\tint h_n$ are \emph{in involution}.
Note that (C4) implies that element $\frac{\delta h_n}{\delta u}$
span an infinite dimensional subspace of $\mc V^\ell$.
The converse holds provided that either $H$ or $K$ is non-degenerate.

Suppose we have a bi-Hamiltonian equation \eqref{20120124:eq5},
associated to the compatible non-local Poisson structures $H,K$
and the Hamiltonian functionals $\tint h_0,\tint h_1$,
in the sense of equation \eqref{20140117:eq1}.
The \emph{Lenard-Magri scheme of integrability}
consists in finding
sequences $\{\tint h_n\}_{n\in\mb Z_+}\subset\mc V/\partial\mc V$ 
and $\{P_n\}_{n\in\mb Z_+}\subset\mc V^\ell$,
starting with $P_0=P$ and the given Hamiltonian functionals $\tint h_0,\tint h_1$,
satisfying the following recursive relations:
\begin{equation}\label{20130604:eq7}
\frac{\delta h_{n}}{\delta u}\ass{H}P_n\ass{K}\frac{\delta h_{n+1}}{\delta u}
\qquad
\text{ for all } n\in\mb Z_+
\,.
\end{equation}
In this case,
we have the corresponding bi-Hamiltonian hierarchy
\begin{equation}\label{20130604:eq6}
\frac{du}{dt_n}=P_n\,\in\mc V^\ell
\,\,,\,\,\,\,
n\in\mb Z_+
\,,
\end{equation}
all Hamiltonian functionals $\tint h_n,\,n\in\mb Z_+$,
are integrals of motion for all equations of the hierarchy,
and they are in involution with respect to both non-local Poisson structures $H$ and $K$,
and all commutators $[P_m,P_n]$ are zero, provided that one of the non-local Poisson 
structures $H$ or $K$ is local (see \cite[Sec.7.4]{DSK13}).
Hence, in this situation \eqref{20130604:eq6} is an integrable hierarchy
of compatible evolution equations,
provided that condition (C4) holds.

%%%
\subsection{A bi-Hamiltonian structure and integrability for the DNLS equation}\label{sec:1.5}

Let $\mc V=\mb C[a^{(n)},b^{(n)}\mid n\in\mb Z_+]$ be the algebra of differential
polynomials in two variables $a$ and $b$.
Sometimes we will also use the notation $a'=a^{(1)}$, $a''=a^{(2)}$ and so on
(and similarly for the $b^{(n)}$'s).

Let $H,K\in\Mat_{2\times2}\mc V((\partial^{-1}))$ be pseudodifferential operators
with coefficients in $\mc V$ defined as follows:
$$
H=
\begin{pmatrix}
\partial &0\\
0 & \partial
\end{pmatrix}
\quad
\text{and}
\quad
K=
\begin{pmatrix}
2\beta b\partial^{-1}\circ b &-1-2\beta b\partial^{-1}\circ a\\
1-2\beta a\partial^{-1}\circ b& 2\beta a\partial^{-1}\circ a
\end{pmatrix}\,,
$$
where $\beta\in\mb C$.
Note that $H(\partial)\in\Mat_{2\times2}\mc V[\partial]$ is in fact a differential operator.

The following result have been proved in \cite{DSK13}.
\begin{theorem}\phantomsection\label{thm:integrability_DNLS}
\begin{enumerate}[(a)]
\item
There exist $A(\partial),B(\partial)\in\Mat_{2\times 2}\mc V[\partial]$, with $B(\partial)$
non-degenerate, such that
$K=A(\partial)B(\partial)^{-1}$. Explicitly:
$$
A(\partial)=
\begin{pmatrix}
-\frac ba &-\frac1a\partial\circ a-2\beta ab\\
1 & 2\beta a^2
\end{pmatrix}
\quad
\text{and}
\quad
B(\partial)=
\begin{pmatrix}
1 & 0\\
\frac ba& \frac1a\partial\circ a
\end{pmatrix}\,.
$$
\item
$(H,K)$ is a bi-Poisson structure on $\mc V$.
\item
There exist infinite sequences $\{\tint h_n\}_{n\in\mb Z_+}\subset\mc V/\partial\mc V$
and $\{P_{n}\}_{n\in\mb Z_+}\subset\mc V^2$ such that the Lenard-Magri recursive relations
\eqref{20130604:eq7} hold.
\item
$\ord\left(\frac{\delta h_n}{\delta u}\right)=n$, for every $n\in\mb Z_+$.
In particular, since $H$ is non-degenerate,
all the elements $\tint h_n$'s and $P_n$'s are linearly independent
(see Section \ref{sec:1.4}).
\end{enumerate}
In conclusion, by the discussion in Section \ref{sec:1.4}, we get an integrable hierarchy of
bi-Hamiltonian equations \eqref{20130604:eq6} and all the Hamiltonian functionals
$\tint h_n$, $n\in\mb Z_+$, are integrals of motion for all equations of the hierarchy.
\end{theorem}
The first few elements in the series of the integrals of motion are
\begin{equation}\label{20141031:eq1}
\tint h_0
=\frac12 \tint \left(a^2+b^2\right)\,,
\qquad
\tint h_1=\tint \Big(ab'+\frac{\beta}{4}(a^2+b^2)^2\Big)
\,.
\end{equation}
The corresponding Hamiltonian equations, given by \eqref{20130604:eq6}, are
$$
\left\{
\begin{array}{l}
\displaystyle{
\frac{da}{dt_0}=a'
}\\
\displaystyle{
\frac{db}{dt_0}=b'
}
\end{array}
\right.
\,,
\qquad
\left\{
\begin{array}{l}
\displaystyle{
\frac{da}{dt_1}=b''+\beta\left(a(a^2+b^2)\right)'
}
\\
\displaystyle{
\frac{db}{dt_1}=-a''+\beta\left(b(a^2+b^2)\right)'
}
\end{array}\right.
\,.
$$
Let us write $\psi=a+ib$. Then, the first non-trivial equation of the hierarchy is
the DNLS equation:
$$%\begin{equation}\label{eq:zurich_DNLS}
i\frac{d\psi}{dt_1}=\psi''+i\beta\left(\psi|\psi|^2\right)'
\,.
$$%\end{equation}

Let us consider $\beta\in\mb C$ as a formal parameter, and let us naturally extend the notion of polynomial
degree and differential degree of $\mc V$ to the field of fractions $\mc K$ and to $\mc K^2$.
The following result is a consequence of the Lenard-Magri recursive relations
\eqref{20130604:eq7} and the
explicit form of the differential operators $A$ and $B$.
\begin{proposition}\label{20141102:prop1}
For every $n\in\mb Z_+$, the variational derivatives $\frac{\delta h_n}{\delta u}$'s are polynomials
in $\beta$ (with coefficients in $\mc V^2$) of order $n$.
Let us write
$$%\begin{equation}\label{20141102:eq1}
\frac{\delta h_n}{\delta u}
=\sum_{k=0}^n\left(\frac{\delta h_n}{\delta u}\right)_{k}\beta^k
\,.
$$%\end{equation}
Then, for every $0\leq k\leq n$, we have
$$
\ord\left(\frac{\delta h_n}{\delta u}\right)_{k}
=n-k
\,.
$$
Moreover, the components of
$\left(\frac{\delta h_n}{\delta u}\right)_{k}$ are homogeneous polynomials with respect to the polynomial
grading (respectively, differential grading) of degree:
$$
\deg\left(\frac{\delta h_n}{\delta u}\right)_{k}=2k+1
\qquad
\Big(\text{respectively, }
\dd\left(\frac{\delta h_n}{\delta u}\right)_{k}
=n-k
\Big)
\,.
$$
\end{proposition}
\begin{proof}
The fact that the variational derivatives $\frac{\delta h_n}{\delta u}$'s are polynomials
in $\beta$ (with coefficients in $\mc V^2$) of order $n$ is true for $n=0,1$ using equation
\eqref{20141102:eq2} and the definition of variational derivative \eqref{20141031:eq2}.
Let us assume that $\frac{\delta h_n}{\delta u}$ has order $n$ as a polynomial in $\beta$, and let us write 
explicitly the Lenard-Magri recursion relations \eqref{20130604:eq7}
using the formulas for the differential operators $A$ and $B$. We get the following system of equations
\begin{equation}\label{20141102:eq3}
\left\{
\begin{array}{l}
\displaystyle{
\partial(a g)
=-a\partial\frac{\delta h_n}{\delta a}-b\partial\frac{\delta h_n}{\delta b}
}
\\
\displaystyle{
\frac{\delta h_{n+1}}{\delta a}
=\partial\frac{\delta h_{n}}{\delta b}-2\beta a^2 g
}
\\
\displaystyle{
\frac{\delta h_{n+1}}{\delta b}
=-\partial\frac{\delta h_{n}}{\delta a}-2\beta ab g
\,,}
\end{array}
\right.
\end{equation}
where $g\in\mc K$ and $\frac{\delta h_{n+1}}{\delta a},\frac{\delta h_{n+1}}{\delta b}\in\mc V$
have to be determined (we know the system can be solved by Theorem \ref{thm:integrability_DNLS}(b)).
From the first equation in \eqref{20141102:eq3} and inductive assumption,
it follows that $g$ is a polynomial of order $n$ in $\beta$.
Then, by the second and third equation in  \eqref{20141102:eq3}, it follows that
$\frac{\delta h_{n+1}}{\delta u}$ is a polynomial of order $n+1$ in $\beta$.

Moreover, by Theorem 7.15(c) in \cite{DSK13}, we have that
$\ord\left(\frac{\delta h_{n+1}}{\delta u}\right)=\ord(P_n)$.
Recall that $P_n=H\left(\frac{\delta h_n}{\delta u}\right)$.
Hence, equating the orders of the coefficients of powers of $\beta$ we get
$$
\ord\left(\frac{\delta h_{n+1}}{\delta u}\right)_{k}
=\ord\left(H\left(\frac{\delta h_n}{\delta u}\right)_{k}\right)=n+1-k
\,.
$$
In the last equality we used the fact that $\partial\mc V_n\subset\mc V_{n+1}$.
The last part of the proposition follows by a simple inductive argument using equations 
\eqref{20141102:eq4} and \eqref{20141102:eq3}.
\end{proof}
\begin{remark}\label{20141102:rem1}
By the first part of Proposition \ref{20141102:prop1}, we can write $h_{n}$ as a polynomial in $\beta$.
By the second part, using the definition of variational derivative and equation
\eqref{20141102:eq4}, we get that
\begin{equation}\label{20141102:eq5}
h_{n}=\sum_{k=0}^{n} h_{n,k}\beta^k
\,,
\end{equation}
where $h_{n,k}\in\mc V$ are homogeneous differential polynomials such that
$\deg(h_{n,k})=2k+2$ and $\dd(h_{n,k})=n-k$.
\end{remark}

%%%
\subsection{Explicit structure of the integrals of motion of the DNLS equation}\label{sec:1.6}

Let us define a sequence $\{\xi_n\}_{n\in\mb Z_+}\subset\mc V^2$ as follows:
\begin{equation}\label{20141102:eq2}
\xi_0=\begin{pmatrix}
a\\b
\end{pmatrix}
\,,
\qquad
\xi_1=\begin{pmatrix}
b'+\beta a(a^2+b^2)\\
-a'+\beta b(a^2+b^2)
\end{pmatrix}
\,,
\end{equation}
and, for $n\geq1$, we set
$$
\xi_{2n}=(-1)^n
\begin{pmatrix}
a^{(2n)}-(2n+1)\beta(a^2+b^2)b^{(2n-1)}+r_{2n}^a
\\
b^{(2n)}+(2n+1)\beta(a^2+b^2)a^{(2n-1)}+r_{2n}^b
\end{pmatrix}
\,,
$$
$$
\xi_{2n+1}=(-1)^n\!
\begin{pmatrix}
\!b^{(2n+1)}
\!+
2\beta a(aa^{(2n)}\!+bb^{(2n)})
+
(2n\!+\!1)\beta(a^2+b^2)a^{(2n)}
\!+\!
r_{2n+1}^a
\\
\!-a^{(2n+1)}
\!+
2\beta b(aa^{(2n)}\!+bb^{(2n)})
+
(2n\!+\!1)\beta(a^2+b^2)b^{(2n)}
\!+\!
r_{2n+1}^b
\end{pmatrix}
\,,
$$
where $r_{2n}^x\in\mc V_{2n-2}$ and $r_{2n+1}^x\in\mc V_{2n-1}$, for $x=a$ or $b$.
\begin{lemma}\label{lem:conti1}
Let us denote $\xi_n=\begin{pmatrix}\xi_n^a\\ \xi_n^b\end{pmatrix}\in\mc V^2$, for every $n\in\mb Z_+$.
Then we have:
\begin{enumerate}[(a)]
\item
\begin{align*}
&a\xi_{2n+1}^b-b\xi_{2n+1}^a
-(-1)^{n+1}\partial\left(aa^{(2n)}+bb^{(2n)}\right)
\\
&-(-1)^n(2n+1)\beta\partial\left(
(a^2+b^2)(ab^{(2n-1)}-a^{(2n-1)}b)\right)\in\mc V_{2n-1}
\,.
\end{align*}
\item
\begin{align*}
&a\xi_{2n}^b-b\xi_{2n}^a
-(-1)^{n}\partial\left(ab^{(2n-1)}-a^{(2n-1)}b\right)
\\
&-(-1)^n(2n+1)\beta\partial\left(
(a^2+b^2)(aa^{(2n-2)}+bb^{(2n-2)})\right)\in\mc V_{2n-2}
\,.
\end{align*}
\item
\begin{align*}
&\xi_{2n+1}^a-\partial\xi_{2n}^b
-(-1)^{n}2\beta a\left(aa^{(2n)}+bb^{(2n)}\right)
\in\mc V_{2n-1}\,,
\\
&\xi_{2n+1}^b+\partial\xi_{2n}^a
-(-1)^{n}2\beta b\left(aa^{(2n)}+bb^{(2n)}\right)
\in\mc V_{2n-1}\,.
\end{align*}
\item
\begin{align*}
&\xi_{2n}^a-\partial\xi_{2n-1}^b
-(-1)^{n+1}2\beta a\left(ab^{(2n-1)}b-a^{(2n-1)}b\right)
\in\mc V_{2n-2}\,,
\\
&\xi_{2n}^b+\partial\xi_{2n-1}^a
-(-1)^{n+1}2\beta b\left(ab^{(2n-1)}-a^{(2n-1)}b\right)
\in\mc V_{2n-2}
\,.
\end{align*}
\end{enumerate}
\end{lemma}
\begin{proof}
Straightforward.
\end{proof}
Let us also define a sequence $\{P_n\}_{n\in\mb Z_+}\subset\mc V^2$ as follows:
$$%\begin{equation}\label{20141101:eq1}
P_n=H\xi_n=
\begin{pmatrix}
\partial\xi_n^a
\\
\partial\xi_n^b
\end{pmatrix}
\,.
$$%\end{equation}
\begin{lemma}\label{lem:conti2}
For every $n\in\mb Z_+$, there exists $F_{n}\in\mc K^2$ such that:
\begin{enumerate}[(a)]
\item
$P_{2n}-AF_{2n}\in\mc V_{2n-1}^2$ and $\xi_{2n+1}-BF_{2n}\in\mc V_{2n-1}^2$;
\item
$P_{2n+1}-AF_{2n+1}\in\mc V_{2n}^2$ and $\xi_{2n+2}-BF_{2n+1}\in\mc V_{2n}^2$.
\end{enumerate}
\end{lemma}
\begin{proof}
For every $n\in\mb Z_+$, let us consider
$$
F_n=
\begin{pmatrix}
\xi_{n+1}^a
\\
f_n+g_n
\end{pmatrix}
\in\mc K^2\,,
$$
where
\begin{align*}
&f_{2n}
=(-1)^{n+1}\frac{aa^{(2n)}+bb^{(2n)}}{a}
+(-1)^n(2n+1)\beta\frac{(a^2+b^2)(ab^{(2n-1)}-a^{(2n-1)}b)}{a}
\,,
\\
&f_{2n+1}
=(-1)^{n+1}\frac{ab^{(2n+1)}\!-a^{(2n+1)}b}{a}
+\!(-1)^{n+1}(2n+3)\beta\frac{(a^2+b^2)(aa^{(2n)}\!+bb^{(2n)})}{a}
\,,
\end{align*}
and $g_n\in\mc V_{n-2}$.
Then, using the definition of the differential operators $A$ and $B$
given by Theorem \ref{thm:integrability_DNLS}(a),
it is straightforward to check that part (a) follows from Lemma \ref{lem:conti1}(a) and (c),
while part (b) follows from Lemma \ref{lem:conti1}(b) and (d).
\end{proof}
\begin{proposition}\label{20141101:prop1}
Let $\{\tint h_n\}_{n\in\mb Z_+}\subset\mc V/\partial\mc V$ be the sequence in
Theorem \ref{thm:integrability_DNLS}.
Then, for every $n\in\mb Z_+$, we have
$$
\frac{\delta h_n}{\delta u}-\xi_n\in\mc V_{n-2}^2
\,.
$$
\end{proposition}
\begin{proof}
By equation \eqref{20141031:eq1} and the definition of variational derivative \eqref{20141031:eq2}
it follows that $\frac{\delta h_n}{\delta u}=\xi_n$, for $n=0,1$.
Hence, by Theorem \ref{thm:integrability_DNLS}(d),
in order to prove the proposition we need to show that the sequence
$\{\xi_n\}_{n\in\mb Z_+}\subset\mc V^2$
satisfies the Lenard-Magri recursive relations \eqref{20130604:eq7} up to elements in $\mc V_{n-2}$.
This follows by definition of the association relation \eqref{20130112:eq2}, the definition of the sequence
$\{P_n\}_{n\in\mb Z_+}\subset\mc V^2$ and Lemma \ref{lem:conti2}(a) and (b).
\end{proof}
\begin{corollary}\label{20141102:cor1}
For every $n\in\mb Z_+$ we can assume that the conserved densities $h_{2n}\in\mc V$, defined by 
Theorem \ref{thm:integrability_DNLS}, have the form:
$$
h_{2n}=\frac{1}{2}\left((a^{(n)})^2+(b^{(n)})^2\right)
+(2n+1)\beta\left(a^2+b^2\right)a^{(n-1)}b^{(n)}
+R_{2n}\,,
$$
where $R_{2n}\in\mc V_{n-1}$.
\end{corollary}
\begin{proof}
It follows by Proposition \ref{20141101:prop1} and the definition of the variational derivative 
\eqref{20141031:eq2}, using the fact that $\partial\mc V_k\subset\mc V_{k+1}$,
for every $k\in\mb Z_+$, and that the variational derivative of a total derivative is zero.
\end{proof}
%

%%%
\subsection{Changing variables}\label{sec:1.7}

Let $\mc V^{\mb C}$ be the algebra of differential polynomials in two variables $\psi$ and $\bar\psi$.
We have a differential algebra isomorphism $\mc V\stackrel{\sim}{\rightarrow}\mc V^{\mb C}$
given on generators by
$$
a=\frac{\psi+\bar\psi}{2}\,,
\qquad
b=\frac{\psi-\bar\psi}{2i}
\,.
$$
Clearly,the inverse map is given by $\psi=a+ib$ and $\bar\psi=a-ib$. (In the usual analytical language,
if $a$ and $b$ are real functions, then we want to consider them as
the real and imaginary parts of the function $\psi$.)

The differential order, the polynomial grading and the differential
grading of $\mc V$ and $\mc V^{\mb C}$ are compatible under this isomorphism.
Hence, all the results in the Section \ref{sec:1.5} 
hold true for
$\frac{\delta h_n}{\delta u}\in\mc (V^{\mb C})^2$
(by an abuse of notation we are denoting with the same symbol an element in $\mc V$ and its image in
$\mc V^{\mb C}$)
Moreover, we can restate Corollary \ref{20141102:cor1} as follows.
\begin{corollary}\label{20141102:cor2}
For every $n\in\mb Z_+$ we can assume that the conserved densities $h_{2n}\in\mc V^\mb C$, defined by 
Theorem \ref{thm:integrability_DNLS}, have the form:
$$
h_{2n}=\frac{1}{2}\psi^{(n)}\bar\psi^{(n)}
+\frac{(2n+1)i}{2}\beta\bar\psi^{(n)}\psi^{(n-1)}\bar\psi\psi
+R_{2n}\,,
$$
where $R_{2n}\in\mc V_{n-1}^\mb C$.
\end{corollary}
\begin{proof}
Clearly, $(a^{(n)})^2+(b^{(n)})^2=\psi^{(n)}\bar\psi^{(n)}$, for every $n\in\mb Z_+$.
Moreover, we have
\begin{align}
\label{eq0}
(a^2+b^2)a^{(n-1)}b^{(n)}
=\frac{i}{4}\left(
\bar\psi^{(n)}\psi^{(n-1)}\bar\psi\psi
-\psi^{(n)}\bar\psi^{(n-1)}\bar\psi\psi\right.
\left.+\bar\psi^{(n)}\bar\psi^{(n-1})\bar\psi\psi
-\psi^{(n)}\psi^{(n-1)}\bar\psi\psi
\right)
\,.
\end{align}
Note that, integrating by parts, we have
\begin{equation}\label{eq1}
\psi^{(n)}\bar\psi^{(n-1)}\bar\psi\psi
=-\psi^{(n-1)}\partial(\bar\psi^{(n-1)}\bar\psi\psi)\bmod\partial\mc V
=(-\bar\psi^{(n)}\psi^{(n-1)}\bar\psi\psi+f)\bmod\partial\mc V
\,,
\end{equation}
where $f\in\mc V^{\mb C}_{n-1}$. Moreover, again using integration by parts,
we have
$$
\psi^{(n)}\psi^{(n-1)}\bar\psi\psi
=-\psi^{(n-1)}\partial(\psi^{(n-1)}\bar\psi\psi)\bmod\partial\mc V
=(-\psi^{(n)}\psi^{(n-1)}\bar\psi\psi+2g)\bmod\partial\mc V
\,,
$$
with $g\in\mc V_{n-1}^{\mb C}$. Then,
\begin{equation}\label{eq2}
\psi^{(n)}\psi^{(n-1)}\bar\psi\psi
=f\bmod\partial\mc V\,.
\end{equation}
Similarly, we get that
\begin{equation}\label{eq3}
\bar\psi^{(n)}\bar\psi^{(n-1)}\bar\psi\psi=h\bmod\partial\mc V\,.
\end{equation}
for some $h\in\mc V_{n-1}^{\mb C}$.
Combining equations \eqref{eq0}, \eqref{eq1}, \eqref{eq2} and \eqref{eq3} the proof is concluded.
\end{proof}
We want to give a description of the conserved densities $h_{2n}\in\mc V^{\mb C}$ which
will be used throughout the rest of the paper.

Let $\widetilde{\mc V}$ be the algebra of differential polynomials in one variable $u$.
Let us denote by
\begin{equation}\label{Hom}
\tilde{}:\mc V^{\mb C}\to\widetilde{\mc V}
\end{equation}
the differential algebra homomorphism defined as follows: given $f\in\mc V^{\mb C}$, we denote
by $\widetilde{f}\in\widetilde{\mc V}$ the differential polynomial obtained by replacing $\psi$ and
$\bar\psi$ by $u$ (and their $n$-th derivatives by $u^{(n)}$). 
Note that $\widetilde{\mc V}$ inherits the polynomial and differential grading of $\mc V^{\mb C}$.

Recall, by Remark \ref{20141102:rem1}, that we can write the conserved densities as in
equation \eqref{20141102:eq5}.
Then, by Corollary \ref{20141102:cor2}, we have that
\begin{equation}\label{Ziurig1}
h_{2n,0}=\frac{1}{2}\psi^{(n)}\bar\psi^{(n)}
\,,
\end{equation}
and
\begin{equation}
\label{Ziurig2}
h_{2n,1}=\frac{(2n+1)i}{2}\bar\psi^{(n)}\psi^{(n-1)}\bar\psi\psi
+\sum_{p\in\widetilde P}c_{2n}(p)p
\,,
\end{equation}
where $c_{k}(p)\in\mb C$ (they can be possibly $0$) and
\begin{equation}\label{eq:PTilde}
\widetilde{P}
=\{p\in\mc V^{\mb C}\mid \widetilde{p}=u^{(n-1)}u^{(n_1)}u^{(n_2)}u^{(n_3)}\,,
n_1+n_2+n_3=n\,,0\leq n_3\leq n_2\leq n_1\leq n-1\}
\,.
\end{equation}

%%%%%%%%%%%%%%%%%%%%%%%%%%%%%%%%%%%%%%%%
\section{Control of the Sobolev Norms}\label{sec:stab}

The goal of this section is to show the
persistence of regularity of small solutions of  
DNLS equation \eqref{eq:dnls}, using the higher Hamiltonians introduced in 
Theorem \ref{thm:integrability_DNLS}.

For every $k\in\mb Z_+$, we denote
$$%\begin{equation}\label{zurich:E_k}
E_k=\tint h_{2k}
\,.
$$%\end{equation}
By equations \eqref{20141102:eq5}, \eqref{Ziurig1}, \eqref{Ziurig2} and Corollary \ref{20141102:cor2} it is possible to write
\begin{equation}\label{Eq:HkLawRepres}
E_k= 
\frac{1}{2}\| \psi \|_{\dot{H}^{k}} 
+
\tint q_{k},
\end{equation}
where
\begin{equation}\label{eq:qkremainderRepres}
q_{k}
:=
\frac{(2k+1)i}{2} \beta \bar{\psi}^{(k)}\psi^{(k-1)} \bar{\psi} \psi +
\beta \sum_{p \in \tilde{P}} c_{2k}(p)  p
+ \sum_{m=2}^{2k}  \beta^{m}  h_{2k,m} \,.
\end{equation}
We recall that dd$(h_{k,m}) = 2k - m$
and $\widetilde{P}$ is defined in \eqref{eq:PTilde} .
\begin{remark}
Note that using Proposition \ref{20141101:prop1} 
(and recalling equation \eqref{20141102:eq5})
it is possible to write
\begin{equation}\nonumber
\tint h_{2k+1} =
\frac{i}{2} \tint \psi^{(k)} \bar{\psi}^{(k+1)}
+
\sum_{m=1}^{2k+1} \beta^{m} \tint h_{2k+1,m} \,.
\end{equation}
Differently from the case of $\tint h_{2k}$, the constant term in $\beta$ of 
the above equation has no
definite sign and, in particular, it does not coincide with $\| \psi \|_{\dot{H}^{k/2}}$. 
\end{remark}

The main result of the section is the following

\begin{proposition}\label{MainPropSecondPart}
Let $k\in\mb Z_+$. For every $0\leq m \leq k$
let us fix $R_m \geq 0$, assuming
$R_{0} \leq \sqrt{\frac{2}{9 |\beta|}}$.
There exists $\mathcal{C} = \mathcal{C}(R_{0},\dots,R_{k},k, |\beta|)$
such that if
$$%\be
| E_m[\psi] | \leq R_m,\quad \mbox{for any} \quad m=0,\dots,k\,,
$$%\ee
then
\be\label{eq:stability}
\|\psi\|_{\dot H^k}\leq \mathcal{C}.
\ee
%\begin{equation}\label{eq:stability}
%\bigcap_{m=0}^{k} 
%\{ 
%\psi \in \dot{H}^k \ \mid \ |E_{m}(\psi)| = R_{m}
%\}
%\subseteq
%\{
%\psi \in \dot{H}^k \ \mid \ \| \psi \|_{\dot{H}^k}\leq \mathcal{C}
%\}\,,
%\end{equation}
%and
%\begin{equation}\label{eq:stabilityBis}
%\bigcap_{m=0}^{k} 
%\{ 
%\psi \in \dot{H}^k \ \mid \ |E_{m}(\psi)| \leq R_{m}
%\}
%\subseteq
%\{
%\psi \in \dot{H}^k \ \mid \ \| \psi \|_{\dot{H}^k} \leq \mathcal{C}
%\}\,,
%\end{equation}

\end{proposition}
To prove Proposition \ref{MainPropSecondPart} we need some
preliminary results.
\begin{lemma}\label{20150518:lem1}
Let $k\geq 2$ and $u \in H^{k-1}$. For $l\geq5$ and $\alpha_{i} \geq 0$ $(i=1,\dots l)$
such that $\alpha_1+\dots +\alpha_l\leq 2k-2$, we have
\begin{equation}\label{20140518:eq1}
\left|\tint u^{(\alpha_1)}\dots u^{(\alpha_l)}\right|
\lesssim 
\|u\|_{H^{k-1}}^l\,.
\end{equation}
\end{lemma}
\begin{proof}
We reorder the terms in the integrand in the l.h.s. of \eqref{20140518:eq1},
such that $\alpha_1\geq\alpha_2\geq\dots\geq\alpha_l$.
Furthermore, using integration by parts, we may assume that
\begin{equation}\label{20150518:eq2}
\alpha_1,\alpha_2\leq k-1\,,
\qquad\text{and}\qquad
\alpha_i\leq k-2\,,\quad i=3,\dots,l\,.
\end{equation}
By the Holder inequality and the first condition
in \eqref{20150518:eq2} we get
\begin{equation}\label{20140518:eq3}
\left|\tint u^{(\alpha_1)}\dots u^{(\alpha_l)}\right|\leq
\| u\|_{\dot{H}^{k-1}}^2\prod_{i=3}^l\|u^{(\alpha_i)}\|_{L^\infty}\,.
\end{equation}
Using the embedding $H^1\hookrightarrow L^\infty$
and the second condition in \eqref{20150518:eq2} we have (for all $i=3,\dots,l$):
\begin{equation}\label{20140518:eq4}
\|u^{(\alpha_i)}\|_{L^\infty}
\lesssim
\|u^{(\alpha_i)}\|_{H^1}
\leq
\|u\|_{H^{k-1}}\,.
\end{equation}
The inequality \eqref{20140518:eq1} follows combining the inequalities
\eqref{20140518:eq3} and \eqref{20140518:eq4}.
\end{proof}
\begin{lemma}\label{20150518:lem2}
Let $k\geq2$ and $u \in H^{k-1}$.
Let also $\alpha_1\geq\alpha_2\geq\alpha_3\geq\alpha_4\geq0$
be such that $\alpha_1+\alpha_2+\alpha_3+\alpha_4=2k-1$.
For $\alpha_1=k-1$ and $\alpha_2,\alpha_3,\alpha_4\leq k-1$, we have
$$
\left|\tint u^{(k-1)}u^{(\alpha_2)}u^{(\alpha_3)} u^{(\alpha_4)}\right|
\lesssim \|u\|_{H^{k-1}}^4\,.
$$
\end{lemma}
\begin{proof}
%proof of part(b)
Same as the proof of Lemma \ref{20150518:lem1}.
\end{proof}
\begin{lemma}\label{20150518:lem3}
Let $k\geq2$ and let $u \in H^{k}$. Then 
$$%\begin{equation}
\left|\tint u^{(k)}u^{(k-1)}u^2\right|
\leq \varepsilon \|u\|^2_{\dot{H}^{k}}
+
C(\varepsilon) \|u\|^6_{H^{k-1}}\,,
$$%\end{equation}
for all $\varepsilon > 0$.
%$$
%\left|\tint u^{(k)}u^{(k-1)}u^2\right|
%\leq \frac{\varepsilon^2}{2}\|u\|^2_{\dot{H}^{k}}+\frac{\mathcal{C}}{2\varepsilon^2}\|u\|%^6_{H^{k-1}}\,.
%$$
\end{lemma}
\begin{proof}
By using the Holder inequality 
and the embedding
$H^1\hookrightarrow L^\infty$
we get
\begin{equation}\nonumber
\left|\tint u^{(k)}u^{(k-1)}u^2\right|
\leq 
\| u\|_{\dot{H}^{k}}\|u\|_{\dot{H}^{k-1}}\|u\|_{L^\infty}^2
\lesssim 
\| u\|_{\dot{H}^{k}}\|u\|_{H^{k-1}}^3\,.
\end{equation}
The proof is concluded by applying the Young inequality in the last expression.
\end{proof}
\begin{corollary}\label{20140519:cor1}
Let $k\geq2$ and $c(k)>0$ a constant depending only by $k$.
For every $\psi \in H^{k}$ and $\e>0$, we have
\begin{equation}\label{eq:qk_estimate}
\left|\tint q_k(\psi)\right|
\leq c(k)\varepsilon 
\|\psi\|_{\dot{H}^{k}}^2+\mathcal{C}\,,
\end{equation}
where $\mathcal{C}=\mathcal{C}(\|\psi\|_{\dot{H}^{0}},
\|\psi\|_{\dot{H}^{k-1}},\varepsilon, k, |\beta|)$.
\end{corollary}
\begin{proof}
Let us focus on the representation
\eqref{eq:qkremainderRepres}:
\begin{equation}\nonumber
q_{k}(\psi)
=
\frac{(2k+1)i}{2} \beta \bar{\psi}^{(k)}\psi^{(k-1)} \bar{\psi} \psi +
\beta \sum_{p \in \tilde{P}} c_{2k}(p)  p
+ \sum_{m=2}^{2k}  \beta^{m}  h_{2k,m} \, .
\end{equation}
%$$
%\tint q_k
%=\sum_{n=0}^{2k-1}\beta^{2k-n}\tint q_{k,n}\,,
%$$
%where $q_{k,n}\in\mc V$  are such that
%$\deg q_{k,n}=2(2k-n+1)$ and $\dd q_{k,n}=n$.
The Lemma \ref{20150518:lem1} and the fact that $|\psi|=|\bar\psi|$
allow us to bound (through the homomorphism defined in \eqref{Hom})
\begin{equation}\label{20150520:eq1}
| \tint h_{k,m} | \leq \mathcal{C}(\|\psi\|_{\dot{H}^0},\|\psi\|_{\dot{H}^{k-1}},k)\,,
\end{equation}
for all $m=2 ,\dots, 2k$. Similarly, Lemma \ref{20150518:lem2}
implies
\begin{equation}\label{20150520:eq1Bis}
| \tint p | \leq \mathcal{C}(\|\psi\|_{\dot{H}^0},\|\psi\|_{\dot{H}^{k-1}},k)\,,
\end{equation}
for all $p \in \widetilde{P}$.
%On the other hand, by the explicit form of $\tint q_{k,2k-1}$ given in equation
%\eqref{eq:q_2k-1}, using Lemmas \ref{20150518:lem2}
Finally, Lemma \ref{20150518:lem3} gives
\begin{equation}\label{20150520:eq2}
| \tint \bar{\psi}^{(k)} \psi^{(k-1)} \bar{\psi} \psi |
\leq \varepsilon \| u \|^2_{\dot{H}^{k}}
+
C(\varepsilon) \|u\|^6_{H^{k-1}}\,.
\end{equation}

Combining the equation \eqref{eq:qkremainderRepres}, the inequalities
\eqref{20150520:eq1}, \eqref{20150520:eq1Bis} and \eqref{20150520:eq2}, the estimate
\eqref{eq:qk_estimate} follows.
\end{proof}
\begin{lemma}\label{20140518:lem4}
Let $\psi \in H^{1}$ and 
let us denote $R_0 = \| \psi \|_{L^2}$.
%
%Then for every $\varepsilon>0$ we have
Then
$$%\begin{equation}\label{20140519:eq1}
\left| \tint \bar{\psi}' \bar{\psi} \psi^{2} \right|
\leq
\frac{3}{2} \| \psi \|_{\dot{H}^1}^2 R_0^2
+ \frac{1}{8 \pi^2} R_0^4 \, .
$$%\end{equation}
\end{lemma}
\begin{proof}
By the H\"older inequality
\begin{equation}\label{20140519:eq2}
\left| \bar{\psi}' \bar{\psi} \psi^{2} \right|\leq
\| \psi \|_{\dot{H}^1}\|\psi^3\|_{L^2}
=\| \psi \|_{\dot{H}^1}\|\psi\|_{L^6}^3\,.
\end{equation}
Using the Gagliardo--Nirenberg inequaltity \eqref{eq:GN} we get
\begin{equation}\label{20140519:eq3}
\| \psi \|_{\dot{H}^1}\| \psi \|_{L^6}^3
\leq  \| \psi \|_{\dot{H}^1}^2\| \psi \|_{L^2}^2
+ \frac{1}{2\pi} \| \psi \|_{\dot{H}^1} \| \psi \|_{L^{2}}^3 
=
 \| \psi \|_{\dot{H}^1}^2 R_{0}^2
+ \frac{1}{2 \pi} \| \psi \|_{\dot{H}^1} R_{0}^3
\, .
\end{equation}
Furthermore, using the Young inequality we have
\begin{equation}\label{20140519:eq4}
\frac{1}{2\pi} \| \psi \|_{\dot{H}^1} R_0^3
\leq  \frac{1}{2} \| \psi\|_{\dot{H}^1}^2 R_0^2
%+ C R_0^4
+ \frac{1}{8 \pi^2} R_{0}^{4} \,.
\end{equation}
Combining \eqref{20140519:eq2}, \eqref{20140519:eq3} and
\eqref{20140519:eq4} the proof follows.
\end{proof}

Now we are ready to prove Proposition \ref{MainPropSecondPart}.

\begin{proof}[Proof of Proposition \ref{MainPropSecondPart}]
We prove \eqref{eq:stability} by induction on $k$.
For $k=0$, there is nothing to prove, since 
$ E_0(\psi) = 1/2\| \psi \|^{2}_{L^{2}}$ (equation \eqref{Eq:Legge0}).

For $k=1$, by equation 
\eqref{Eq:Legge1} we can write
$ E_1(\psi)= 1/2\|\psi\|_{\dot{H}^1}^2+\tint q_1(\psi)$, where
$$
\tint q_1(\psi)=\frac{3i}{4} \beta \tint \bar{\psi}' \bar{\psi} \psi^{2}
+\frac{\beta^2}{4} \|\psi\|_{L^6}^6\,.
$$
Hence
\begin{equation}\label{20140519:eq6}
\frac12 \|\psi\|^{2}_{\dot{H}^1}
= E_1(\psi)-\tint q_1(\psi)
\leq  E_1(\psi) - \frac{3i}{4} \beta \tint \bar{\psi}' \bar{\psi} \psi^{2} \,.
\end{equation}
%Using the fact that $|\psi|=|\bar\psi|$ and 
By Lemma \ref{20140518:lem4} and choosing 
$R_0 \leq \sqrt{\frac{2}{9 | \beta |}}$ we obtain
%\begin{equation}
%\left|\int\psi^2\bar\psi\bar\psi'\right|
%\leq\left(\mathcal{C}_1+\frac{\varepsilon^2}{2}
%\right)R_0^2\|\psi\|_{\dot{H}^1}^2
%+\frac{\mathcal{C}_2^2R_0^4}{2\varepsilon^2}\,.
%\end{equation}
\begin{equation}\label{20140519:eq7}
\Big| 
\frac{3i}{4} \beta \tint \bar{\psi}' \bar{\psi} \psi^{2}  
\Big| \leq \frac{1}{4} \| \psi \|^{2}_{\dot{H}^{1}} + \frac{3}{32}R_{0}^{4} \, .
\end{equation}
Thus,
by \eqref{20140519:eq6} and \eqref{20140519:eq7},
it follows that
$$%\be\label{eq:boundH1}
\frac{1}{4} \| \psi \|_{\dot{H}^1}^2
\leq
  | E_{1} | + \frac{3}{32}R_{0}^{4}
=: \mathcal{C}(R_0,R_1)\,.
%% NON E? \mathcal{C}(R_0,R_1,|\beta|) ??? \mathcal{C}ONTROLLA
$$%\ee
This proves \eqref{eq:stability} in the case $k=1$.
Let us assume that equation \eqref{eq:stability} holds for $k \geq 2$,
namely
$$%\begin{equation}
\| \psi \|_{\dot{H}^{k}}
\leq \mathcal{C}(R_{0}, \dots , R_{k}, k, |\beta |) \, ,
$$%\end{equation}
and let us 
show that it holds for $k+1$. 
By equation \eqref{Eq:HkLawRepres} and Corollary \ref{20140519:cor1}
we have
\begin{equation}\label{20140519:eq8}
\begin{split}
\frac{1}{2}\|\psi\|^2_{\dot{H}^{k+1}}
&\leq \left| E_{k+1}(\psi)\right| - \tint q_{k+1}(\psi) 
\\
&\leq R_{k+1}+c(k) \varepsilon \|\psi\|_{\dot{H}^{k+1}}^2
+\mathcal{C}(\|\psi\|_{\dot{H}^0},\|\psi\|_{\dot{H}^k},\varepsilon,k,|\beta|)\,.
\end{split}
\end{equation}
On the other by the inductive assumption 
we have
$$
\mathcal{C}(\|\psi\|_{\dot{H}^0},\|\psi\|_{\dot{H}^k},\varepsilon,k,|\beta|)
=\mathcal{C}(R_0,\dots,R_k,\varepsilon,k,|\beta|)\,.
$$
Hence, from \eqref{20140519:eq8}, choosing $\varepsilon \leq 1/ 4 c(k)$,
we get
$$
\frac{1}{4} \|\psi\|^{2}_{\dot{H}^{k+1}}\leq \mathcal{C}(R_0,\dots,R_k,R_{k+1},k+1,|\beta|)\,,
$$
thus proving the equation \eqref{eq:stability} and concluding the proof. 
\end{proof}

%%%%%%%%%%%%%%%%%%%%%%%%%%%%%%%%%%%%%%%%%%%%%%%%%%%%%%%
\section{Convergence of the Integrals of Motion}\label{sec:CON}

In this section we study the convergence of $G_{k,N}(\psi)$
defined in \eqref{eq:G_k} with respect to the Gaussian measure $\gamma_{k}$. 
The main result is given by the following 

\begin{proposition}\label{prop:conv-mes}
Let $k\geq2$ and $1\leq m\leq k$. Then $\int q_m(\psi_N)$ converges in measure to $\int q_m(\psi)$ w.r.t. the Gaussian measure $d \g_k$. Furthermore, if 
$1\leq m<k$, then $E_m(\psi_N)$ converges in measure to $E_m(\psi)$ w.r.t. $\g_k$.
\end{proposition}

As a consequence, by composition and multiplication of continuos functions, we obtain 

\begin{corollary}\label{20130520:cor1}
The sequence $G_{k,N}(\psi)$ converges in measure,
with respect to $\g_k$, as $N\to\infty$, to
a function which we (already) denoted $G_k(\psi)$.
\end{corollary}

We split the proof of Proposition \ref{prop:conv-mes} in several steps.

\begin{lemma}\label{20150519:lem2}
Let $k\geq2$, and let $\alpha_1\geq\alpha_2\geq\alpha_3\geq\alpha_4\geq0$
be such that $\alpha_1+\alpha_2+\alpha_3+\alpha_4=2k-1$.
For $\alpha_1=k-1$ and $\alpha_2,\alpha_3,\alpha_4\leq k-1$,
we have
$$
\lim_{N\to\infty}
\tint u_N^{(k-1)}u_N^{(\alpha_2)}u_N^{(\alpha_3)} u_N^{(\alpha_4)}
=\tint u^{(k-1)}u^{(\alpha_2)}u^{(\alpha_3)} u^{(\alpha_4)}
\, ,
$$
almost everywhere with respect to the measure $\g_k$.
\end{lemma}
\begin{proof}
We have
\begin{equation}\nonumber
 | \tint u_N^{(k-1)}u_N^{(\alpha_2)}u_N^{(\alpha_3)} u_N^{(\alpha_4)}
- \tint u^{(k-1)}u^{(\alpha_2)}u^{(\alpha_3)} u^{(\alpha_4)} | 
\leq A_{1} + A_{2},
\end{equation}
where 
\begin{equation}\nonumber
A_{1} := | \tint ( u_N^{(k-1)} - u^{(k-1)}) u_N^{(\alpha_2)}u_N^{(\alpha_3)} u_N^{(\alpha_4)} |,
\quad 
A_{2} := |  \tint  u^{(k-1)} ( u_N^{(\alpha_2)}u_N^{(\alpha_3)} u_N^{(\alpha_4)} - u^{(\alpha_2)}u^{(\alpha_3)} u^{(\alpha_4)}  ) |
\end{equation}
by using the embedding $H^{1}  \hookrightarrow L^{\infty}$ and $u_{N} \rightarrow u$ in $\dot{H}^{k-1}$,  $\g_k$-a.s. we
immediately see that $A_{1} \rightarrow 0$, $\g_k$-a.s..
Then we notice that
\begin{equation}\nonumber
A_{2}
\leq 
B_{1}  + B_{2}
\end{equation}
where
\begin{equation}\nonumber
B_{1} :=  |  \tint u^{(k-1)}  (u_N^{(\alpha_2)} - u^{(\alpha_2)}) u_N^{(\alpha_3)} u_N^{(\alpha_4)}  |,
\quad
B_{2} :=  |  \tint u^{(k-1)}  u^{(\alpha_2)} ( u^{(\alpha_3)} u^{(\alpha_4)}  - u_N^{(\alpha_3)} u_N^{(\alpha_4)} )  |
\end{equation}
and as before $B_{1} \rightarrow 0 $, $\g_k$-a.s..
We finally notice that
\begin{equation}\nonumber
B_{2} \leq C_{1}  + C_{2}
\end{equation}
where
\begin{equation}
C_{1} := | \tint u^{(k-1)}  u^{(\alpha_2)} ( u^{(\alpha_3)} - u_N^{(\alpha_3)}) u^{(\alpha_4)}   |,
\quad
C_{2} := | \tint u^{(k-1)}  u^{(\alpha_2)}  u_N^{(\alpha_3)} (u^{(\alpha_4)} - u_N^{(\alpha_4)})    |
\end{equation}
and as before both $C_{1}, C_{2} \rightarrow 0$, $\g_k$-a.s., which completes the proof. 

\end{proof}

\begin{lemma}\label{20150519:lem1}
For $k\geq2$, $l\geq5$, and $\alpha_{i} \geq 0$ ($i=1, \dots ,l$) 
such that
$0\leq\alpha_1+\dots+\alpha_l\leq 2k-2$, we have
$$
\lim_{N\to\infty}\tint u_N^{(\alpha_1)}\dots u_N^{(\alpha_l)}
=\tint u^{(\alpha_1)}\dots u^{(\alpha_l)}\,,
$$
almost everywhere with respect to the measure $\g_k$.
%
%Moreover,
%$$
%\lim_{N\to\infty}\|\tint u_N^{(\alpha_1)}\dots u_N^{(\alpha_l)}
%-\tint u^{(\alpha_1)}\dots u^{(\alpha_l)}\|_{L^2(\g_k)}=0\,.
%$$
\end{lemma}
\begin{proof}
%Let $u\in \bigcap_{s} H^{s}$, $s<k-\frac12$.
%
As in the proof of Lemma \ref{20150518:lem1}, by reordering and integration by parts we can 
reduce to the case
\begin{equation}\nonumber
\alpha_1,\alpha_2\leq k-1\,,
\qquad\text{and}\qquad
\alpha_i\leq k-2\,,\quad i=3,\dots,l\,.
\end{equation}
Then the proof is the same of Lemma \ref{20150519:lem2}.

\end{proof}

Let $l\in\mb Z_+$. We denote by $S_l$ the group of permutations on $l$ elements.
In the sequel we use the following version of the Wick formula
(we refer to \cite{Ca73} or to \cite{guerra,Simon} for more details).
Let $(m_1,\dots,m_l,n_1,\dots,n_l)\in\mb Z^{2l}$.
Then we have
\be\label{eq:Wick-L}
\mathbb{E}\left[ \prod_{j=1}^l\bar\psi_{m_j}\psi_{n_j} \right]=
\sum_{\sigma \in S_l}\prod_{i=1}^l\frac{\d_{m_i,n_{\sigma(i)}}}{(1+|n_{\sigma(i)}|^{k})^2}\,.
%\sum_{\mu,\nu\in P_4}\prod_{m_j\in\mu}\prod_{n_h\in\nu}\frac{\d_{m_j,n_h}}{(1+|m_j|^{k})(1+|n_h|^{k})}.
\ee
Let us denote by 
\begin{equation}\label{Def:fNk}
f^k_N(\psi):=\tint \psi_N^{(k)}\bar\psi_N^{(k-1)}\bar\psi_N\psi_N \, .
\end{equation}

\begin{proposition}\label{20140519:prop1}
Let $k\geq2$.
The sequence $\{f^k_N\}_{N\in\mb Z_+}$ is a Cauchy sequence in $L^2_{\gamma_k}$, for all $s<k-1/2$. Indeed, for all $N>M\geq1$, we have
$$%\be\label{eq:L2-Cauchy}
\|f^k_M-f^k_N\|_{L^2_{\gamma_k}}
\lesssim \frac{1}{\sqrt{M}} \,.
$$%\ee
\end{proposition}
\begin{proof}
By an explicit computation we get
$$%\begin{equation}\label{20140523:eq1}
f^k_N(\psi)
= i \sum_{A_N}n_1^km_1^{k-1}
\bar\psi_{m_1}\bar\psi_{m_2}\psi_{n_1}\psi_{n_2}\,,
$$%\end{equation}
where
$$%\be\label{eq:A_N}
A_N:=\{(m_1,m_2,n_1,n_2)\in\mb Z^4\mid
|m_i|,|n_i|\leq N\,,\, n_1+n_2 = m_1+m_2\}\,.
$$%\ee
According to our convention, the labels $m_i$ (respectively $n_i$) are associated to the Fourier coefficients of $\bar\psi$ (respectively $\psi$). Moreover we define
$$%\be\label{eq:A_NM}
A_{N,M}:=\{(m_1,m_2,n_1,n_2)\in A_N\,,\,\max(|m_1|,|m_2|,|n_1|,|n_2|)>M\}\,.
$$%\ee
Thus
\begin{equation}\label{takingsquare}
f^k_N(\psi)-f^k_M(\psi)=
i  \sum_{A_{N,M}}n_1^km_1^{k-1}
\bar\psi_{m_1}\bar\psi_{m_2}\psi_{n_1}\psi_{n_2}\,.
\end{equation}
Taking the square of equation \eqref{takingsquare} we get
$$%\be\label{eq:|fn-fm|^2}
|f^k_N(\psi)-f^k_M(\psi)|^2=\sum_{A_{N,M}\times A'_{N,M}}
n_1^km_1^{k-1}m_3^kn_3^{k-1}
%\prod_{j=1}^4(1+|m_j|^k)(1+|n_j|^k)
\prod_{j=1}^4\bar\psi_{m_j}\psi_{n_j}\, ,
$$%\ee
where
\begin{align*}%\be\label{eq:A_N}
&A'_N:=\{(m_3,m_4,n_3,n_4)\in\mb Z^4\mid
|m_i|,|n_i|\leq N\,,\,m_3+m_4=n_3+n_4\}\, ,
\\
%\ee
%\be\label{eq:A'_NM}
&A'_{N,M}:=\{(m_3,m_4,n_3,n_4)\in A'_N\mid \max(|m_3|,|m_4|,|n_3|,|n_4|)>M\}\,.
\end{align*}%\ee
By definition of the measure $\g_k$ we have
\be\label{eq:|f_n-f_m|^2Wick}
\|f^k_M-f^k_N\|_{L^2_{\gamma_k}}^2=\sum_{A_{N,M}\times A'_{N,M}}
n_1^km_1^{k-1}m_3^kn_3^{k-1}
%\prod_{j=1}^4(1+|m_j|^k)(1+|n_j|^k)
\mathbb{E}\left[ \prod_{j=1}^4\bar\psi_{m_j}\psi_{n_j} \right]\,.
\ee
By using the Wick formula \eqref{eq:Wick-L} with $l=4$, equation \eqref{eq:|f_n-f_m|^2Wick}
becomes
\be\label{eq:Wick4}
\|f^k_M-f^k_N\|_{L^2_{\gamma_k}}^2=\sum_{A_{N,M}\times A'_{N,M}}
n_1^km_1^{k-1}m_3^kn_3^{k-1}
\sum_{\sigma \in S_4}\prod_{i=1}^4\frac{\d_{m_i,n_{\sigma(i)}}}{(1+|n_{\sigma(i)}|^{k})^2}\,.
%\sum_{\mu,\nu\in P_4}\prod_{m_j\in\mu}\prod_{n_h\in\nu}\frac{\d_{m_j,n_h}}{(1+|m_j|^{k})(1+|n_h|^{k})}.
\ee
%where $\mu$ (respectively $\nu$) is an quadruple obtained as a permutation of $(m_1,\dots,m_4)$ (respectively $(n_1,\dots,n_4)$), $P_m$ ($P_n$) is the set of all the permutation of the $m$ ($n$) indices and 
%\be
%\mathbb{E}[\bar\psi_{m_j}\psi_{n_h}]=\d_{m_j,n_h}.
%\ee
%As usual in this type of problems this coefficient acts contracting couples of indices. We will indicate a Wick contraction between the index $n_i$ and $m_j$ as $[i,j]$. 

Let us consider the subgroup $G=\{1,(12),(34),(12)(34)\}\subset S_4$
and its action on $S_4$ by left multiplication.
For $X\subset S_4$, we denote by $G\cdot X=\{gx\mid g\in G,x\in X\}$ the orbit of the subset $X$.
We have the following partition of $S_4=W_1\cup W_2\cup W_3$, where
$W_1:=G\cdot\{1\}=G$, $W_2:=G\cdot\{(13),(14),(23),(24)\}$
and $W_3:=G\cdot \{(13)(24)\}$.
%
%$$
%Y(n_1,n_2, m_1, m_2)
%:=\frac{n_{1}^{k} m_{1}^{k-1}}{\prod_{j=1}^2(1+|m_j|^{k})(1+|n_j|^{k})}
%\, .
%$$
Hence, we can further rewrite equation \eqref{eq:Wick4} as follows:
\be\label{eq:Wick5}
\|f^k_M-f^k_N\|_{L^2_{\gamma_k}}^2
=\sum_{i=1}^3\sum_{A_{N,M}^{i}}\sum_{\sigma\in W_i} \frac{n_{1}^{k} n_{\s(1)}^{k-1}n_3^{k-1}n_{\s(3)}^k}{\prod_{j=1}^4(1+|n_j|^{k})^2}
%Y(n_1,n_2,n_{\sigma(1)},n_{\sigma(2)})Y(n_{\sigma(3)},n_{\sigma(4)},n_3,n_4)
\,,
\ee
where the subsets of indices $A^{i}_{N,M}$ will be presented case by case.

We consider the three contributions to the sum in \eqref{eq:Wick5} separately.

%%%
\subsubsection*{First case: \texorpdfstring{$i=1$}{i=1}}
We have
$$
A^1_{N,M}
=
\{(n_1,n_2,n_3,n_4)\in\mb Z^4\mid |n_i|\leq N\,,\max(|n_1|,|n_2|)> M\,,\max(|n_3|,|n_4|)> M\}
\,,
$$
and the contribution to the sum in \eqref{eq:Wick5} is
\be\label{primotermine}
\sum_{A_{N,M}^1}\left(
\frac{n_1^{2k-1}n_3^{2k-1}}{\prod_{j=1}^4(1+|n_j|^k)^2}
+\frac{n_1^{k}n_2^{k-1}n_3^{2k-1}}{\prod_{j=1}^4(1+|n_j|^k)^2}%Y(n_3,n_4,n_3,n_4)\right.
+\frac{n_1^{2k-1}n_3^{k-1}n_4^{k}}{\prod_{j=1}^4(1+|n_j|^k)^2}%+Y(n_1,n_2,n_1,n_2)Y(n_4,n_3,n_3,n_4)
+\frac{n_1^{k}n_2^{k-1}n_3^{k-1}n_4^k}{\prod_{j=1}^4(1+|n_j|^k)^2}%Y(n_1,n_2,n_2,n_1)Y(n_3,n_4,n_3,n_4)
\right)
\,.
\ee
The sum in \eqref{primotermine} is zero. In fact, all the functions involved in the sum are odd functions
with respect to the transformation $n_1\to-n_1$, $n_2\to-n_2$ while
the index set $A^1_{N,M}$ is invariant.

%%%
\subsubsection*{Second case: \texorpdfstring{$i=2$}{i=2}}
In this case we have
$$
A^2_{N,M}
=\{(n_1,n_2,n_3)\in\mb Z^3\mid |n_i|\leq N\,,\max(|n_1|,|n_2|)>M,\,,\max(|n_1|,|n_3|)> M\}.
$$
Similarly to the previous case,
the contribution in the sum \eqref{eq:Wick5} corresponding to a permutation $\sigma\in W_2$
which fixes $1$ (respectively $3$) is zero since the summand is odd with respect to the
transformation $n_1\to-n_1$ (respectively $n_3\to-n_3$) while the index set $A^2_{N,M}$ is invariant.
The summands corresponding to the remaining elements in $W_2$ have the following form
\be\label{4.15}
\sum_{A^2_{N,M}} \frac{n_1^{a_1}}{(1+|n_1|^{k})^4}\frac{n_2^{a_2}}{(1+|n_2|^{k})^2}\frac{n_3^{a_3}}{(1+|n_1|^{k})^2}
\ee
where $a_2,a_3\in \{0,k-1,k\}$, $a_1+a_2+a_3=4k-2$ (hence $2k-2\leq a_{1}\leq4k-2$). So, by a straightforward 
computation, we have (we remind that we are considering $k\geq2$)
 \begin{equation}\label{eq:contrW2}
\eqref{4.15}
\lesssim 
\sum_{\substack{\max(|n_1|,|n_2|)>M, \\ \max(|n_1|,|n_3|)> M}}
\frac{n_1^{a_1}}{(1+|n_1|^{k})^2}\frac{n_2^{a_2}}{(1+|n_2|^{k})^2}\frac{n_3^{a_3}}{(1+|n_3|^{k})^2}
\lesssim \frac{1}{M}.
\end{equation}

%%%
\subsubsection*{Third case: \texorpdfstring{$i=3$}{i=3}}
We have
$$
A^3_{N,M}
=
\{(n_2,n_3,n_4)\in\mb Z^3\mid |n_i|\leq N\,,\max(|n_3+n_4-n_2|,|n_2|,|n_3|,|n_4|)> M\}
\,.
$$
Two summands in \eqref{eq:Wick5}, corresponding to the elements $(13)(24)$
and $(1423)$ in $W_3$, have respectively the following form
\bea
&&\sum_{A^3_{N,M}} \frac{(n_3+n_4-n_2)^{2k}}{(1+|n_3+n_4-n_2|^{k})^2}\frac{1}{(1+|n_2|^{k})^2}\frac{n_3^{2(k-1)}}{(1+|n_3|^{k})^2}\frac{1}{(1+|n_4|^{k})^2}\label{(13)(24)}\,,\\
&&\sum_{A^3_{N,M}} \frac{(n_3+n_4-n_2)^{2k}}{(1+|n_3+n_4-n_2|^{k})^2}\frac{1}{(1+|n_2|^{k})^2}\frac{n_3^{k-1}}{(1+|n_3|^{k})^2}\frac{n_4^{k-1}}{(1+|n_4|^{k})^2}\label{(1423)}\,.
\eea
We can bound these terms as
\bea\label{eq:contrW3-1}
\eqref{(13)(24)}&\lesssim&
\sum_{\max (|n_2|,|n_3|,|n_4|) > M/3}
\frac{1}{(1+|n_2|^{k})^2}\frac{n_3^{2(k-1)}}{(1+|n_3|^{k})^2}\frac{1}{(1+|n_4|^{k})^2}
\lesssim
\frac{1}{M}\,,\label{eq:bound-(13)(24)}\\
\eqref{(1423)}&\lesssim&
\sum_{\max (|n_2|,|n_3|,|n_4|) > M/3}
\frac{1}{(1+|n_2|^{k})^2}\frac{n_3^{k-1}}{(1+|n_3|^{k})^2}\frac{n_4^{k-1}}{(1+|n_4|^{k})^2}
\lesssim
\frac{1}{M^{k}}\,.\label{eq:bound-(1423)}
\eea
The other two terms correspond to $(14)(23)$ and $(1324)$. They can be estimated respectively as
\bea
\sum_{A^{3}_{N,M}} \frac{(n_3+n_4-n_2)^{k}}{(1+|n_3+n_4-n_2|^{k})^2}
\frac{n_2^{k}}{(1+|n_2|^{k})^2}\frac{n_3^{k-1}}{(1+|n_3|^{k})^2}\frac{n_4^{k-1}}{(1+|n_4|^{k})^2}&\lesssim&\frac{1}{M^{k-1}}\label{(14)(23)}
\,,\\
\sum_{A^{3}_{N,M}} \frac{(n_3+n_4-n_2)^{k}}{(1+|n_3+n_4-n_2|^{k})^2}
\frac{n_2^{k-1}}{(1+|n_2|^{k})^2}\frac{n_3^{2(k-1)}}{(1+|n_3|^{k})^2}\frac{1}{(1+|n_4|^{k})^2}&\lesssim&\frac{1}{M^{k-1}}\label{(1324)}
\,.
\eea
%Again, it can be shown that
%\begin{equation}\label{eq:contrW3-2}
%\eqref{ultimotermine}\lesssim 
%\frac{1}{M^{k}}
%\sum_{\max (|n_2|,|n_3|,|n_4|) > M/3}
%\frac{n_2^{k}}{(1+|n_2|^{k})^2}\frac{n_3^{k-1}}{(1+|n_3|^{k})^2}\frac{n_4^{k-1}}{(1+|n_4|^{k})^2}
%\lesssim \frac{1}{M^{2k-1}}.
%\end{equation}
%
In conclusion, recollecting all the contributions given by \eqref{eq:contrW2} and (\ref{eq:bound-(13)(24)}-\ref{(1324)}) , we see immediately that, for $k\geq 2$, we have
\be\label{eq:conv-L2-last}
\|f^k_N-f^k_M\|^2_{L^2_{\gamma_k}} \lesssim \frac{1}{M}\,,
\ee
thus concluding the proof.
\end{proof}
We can extend the estimate \eqref{eq:conv-L2-last} to all the $L^{p}(H^s,\g_k)$-norms, with $p\geq1$.
For $1\leq p< 2$ it is trivial, since $\g_k$ is a probability measure.
For $p>2$ we have to use the properties of the Gaussian measure. For any $r-$linear form $\Psi^r(\psi)$, a direct application of the Nelson hypercontractivity inequality \cite{N73}, as shown for instance in \cite[Theorem I.22]{Simon}, %(see also \cite{TzPTRF}),
yields
$$
\|\Psi^r\|_{L^p_{\g_k}}\leq (p-1)^{\frac r2}\|\Psi^r\|_{L^2_{\g_k}}\,.
$$
This leads us to the following
\begin{corollary}\label{cor:Hyp}
For all $p\geq2$ and $N>M\geq1$, we have
\be\label{eq:Lp-Cauchy}
\|f^k_M(\psi)-f^k_N(\psi)\|_{L^p(H^s,\g_k)}\lesssim\frac {(p-1)^2}{\sqrt{M}}\,.
\ee
\end{corollary}

\begin{corollary}\label{20130520:cor2}
Let $k\geq2$, then
$\tint q_{k,2k-1}(\psi_N)$ converges in measure to $\tint q_{k,2k-1}(\psi)$,
w.r.t. $\g_k$.
\end{corollary}
\begin{proof}
It follows by the explicit form of $\int q_{k,2k-1}$ given in Corollary \ref{20141102:cor2} and by Proposition \ref{20140519:prop1} and Lemma \ref{20150519:lem2}.
\end{proof}

Finally we can prove Proposition \ref{prop:conv-mes}. 
\begin{proof}[Proof of Proposition \ref{prop:conv-mes}]
The explicit form of $\tint q_k$ given by Corollary \ref{20141102:cor2}, 
Lemma \ref{20150519:lem1} and Corollary \ref{20130520:cor2} imply that $\tint q_m(\psi_N)$ 
converges in measure to $\tint q_m(\psi)$ w.r.t. $\g_k$, for $1\leq m\leq k$, $k\geq2$. 
In addition, Proposition \ref{MainPropSecondPart} ensures that as long as $1\leq m<k$
we have $\|\psi_N\|_{H^{m}}\leq C$ $N$-uniformly, thereby it converges to $\|\psi\|_{H^{m}}$ 
 a.e. w.r.t. $\g_k$.
\end{proof}

%%%%%%%%%%%%%%%%%%%%%%%%%%%%%%%%%%%%%%%%%%%%%%%%%%%%%%%
\section{Proof of Theorem \ref{thm:k_even}}\label{sec:fin}

In this section we conclude the proof of Theorem \ref{thm:k_even}.
%We first need some accessory results which follows by Proposition \ref{20140519:prop1}. 
First, we state a useful technical lemma that we borrow from \cite[Proposition 4.5]{TzPTRF}. We report the proof for the sake of completeness:
\begin{lemma}\label{lem:5.1}
Let $(\Omega, \mathcal{S},\mu)$ a finite measure space. If there are $C,r>0$, and an integer $p_0>0$,
such that for every $p\geq p_0$ we have 
$$%\be\label{eq:lim-Lp-bound}
\|F\|_p\leq Cp^r\,,
$$%\ee
then there exist $0<\d<r e^{-1}$ and a constant $L=L(r,\d,p_0)$ such that 
\be\label{eq:exp-p-bound}
\int_{\Omega} d\mu \exp\left[\d\left(\frac{|F|}{C}\right)^{\frac1r}\right]\leq L.
\ee
\end{lemma}
\begin{proof}
We expand
$$
\exp\left[\d\left(\frac{|F|}{C}\right)^{\frac1r}\right]
=\sum_{n\in\mb Z_+} \frac{\d^n}{n!}\left(\frac{|F|}{C}\right)^{n/r}.
$$
Thus
\bea
\int_{\Omega} d\mu \exp\left[\d\left(\frac{|F|}{C}\right)^{\frac1r}\right]&=&\int_{\Omega} dx \sum_{n\in\mb Z_+} \frac{\d^n}{n!}\left(\frac{|F|}{C}\right)^{n/r}\nn\\
&=&\sum_{n\in\mb Z_+} \frac{\d^n}{n!} \frac{\|F\|_{n/r}^{n/r}}{ C^{n/r}}\nn\\
&\leq& \sum_{n< p_0 r} \frac{\d^n}{n!} \frac{\|F\|_{n/r}^{n/r}}{ C^{n/r}}+ \sum_{n\geq p_0r} \frac{\d^n}{n!} \left(\frac nr\right)^n\nn\\
&=:& \sum_{n< p_0r} \frac{\d^n}{n!} \frac{\|F\|_{n/r}^{n/r}}{ C^{n/r}} +L_1(r,\d,p_0)\nn,
\eea
where the constant $L_1(r,\d)$ is finite for $\d<r e^{-1}$. For the finite sum we readily have
$$
\|F\|_{n/r}^{n/r}\leq \|F\|^{n/r}_{p_0}\leq C^{n/r}p_0^n,
$$
hence
$$
\sum_{n< p_0r} \frac{\d^n}{n!} \frac{\|F\|_{n/r}^{n/r}}{ C^{n/r}}\leq\sum_{n< p_0r} \frac{\d^n}{n!} p_0^n=:L_2(r,\d,p_0).
$$
The constant $L_2$ is always finite, so we can set $L=L_1+L_2$ and the assert follows.
\end{proof}

\begin{remark}
The exponent $1/r$ in (\ref{eq:exp-p-bound}) is optimal: the formula remains valid for each $\a\leq 1/r$ and fails otherwise.
\end{remark}

By using Lemma \ref{lem:5.1} and Proposition \ref{20140519:prop1} we can deduce that we have a sub-exponential tail for the convergence in probability of the Cauchy sequence $f^k_N$ defined in
equation \eqref{Def:fNk}.

\begin{lemma}\label{lemma:Cauchy-in-P}
Let $N > M \geq 1$ be integer numbers and $f^k_N$ defined as in \eqref{Def:fNk}. 
Then for any $\l>0$ and $k\geq2$ we have
\be\label{eq:Cauchy-in-P}
\g_k\left(|f^k_N-f^k_M|\geq\l^2 \right)
\lesssim \exp \left( -\frac{\sqrt{2}}{3}\l M^{1/4} \right) .
\ee
\end{lemma}
\begin{proof}
By formula (\ref{eq:Lp-Cauchy}) in Proposition \ref{20140519:prop1} we can apply the Lemma
\ref{lem:5.1} with $F=f^k_N-f^k_M$, $p_0=2$, $r=2$, $C=2 / \sqrt{M}$ and $\d=2/3$. We immediately obtain
$$%\be
\int \g_k(d\psi) \exp\left[\frac 23\left(\frac{|f^k_N-f^k_M|\sqrt{M}}{2}\right)^{1/2}\right]<\infty\,.
$$%\ee
Formula (\ref{eq:Cauchy-in-P}) follows straightforwardly from Markov inequality:
\bea
\g_k\left(|f^k_N-f^k_M|\geq\l^2 \right)&=&\g_k\left(\sqrt{\frac{\sqrt{M}|f^k_N-f^k_M|}{2}}\geq\frac{\l M^{1/4}}{\sqrt{2}} \right)\nn\\
&\leq & 
\exp \left( - \frac23\frac{ \l M^{1/4}}{\sqrt{2}} \right) 
\mathbb{E}\left[
\exp \left( \frac23 \sqrt{\frac{\sqrt{M}|f^k_N-f^k_M|}{2}} \right)
\right]\nn\\
&\leq & \nonumber
 \exp \left( -\frac23\frac{\l M^{1/4}}{\sqrt{2}} \right)
=
 \exp \left( -\frac{\sqrt{2}}{3}\l M^{1/4} \right) \, .
\eea
\end{proof}

%%%%%%%%%%

Now we come to the most important result of this section, namely the integrability of the density $G_{k,N}(\psi)$ w.r.t. the Gaussian measure $\g_k$. More precisely we state:

\begin{proposition}\label{Prop:densita-in-Lp}
Let $\mathcal{C}=\mathcal{C}(R_0,...,R_{k-1},k,|\b|)$ be the constant appearing in Proposition \ref{MainPropSecondPart} and let us take $R_0 \leq \sqrt{\frac{2}{9 |\b|}}$ such that
%\be\label{eq:p0}
$$p_0:=\min \Big( 2 \Big(3(2k+1) | \beta | \sqrt{\mathcal{C}R_0^3} \Big)^{-1}, \Big( 4(2k+1) |\beta| R_0\mathcal{C} \Big)^{-1} \Big)>1.$$
%\ee
Then for any $k\geq2$, $1\leq p<p_0$ and $N \geq \left( \frac{2k+1}{2} | \beta | \right)^{2} R_{0}^{6} \mathcal{C}^{2}$ we have 
$$%\be\label{eq:integrabilityG}
\|G_{k,N}(\psi)\|_{L^p(\g_k)} \leq C < +\infty\, ,
$$%\ee
where $G_{k,N}(\psi)$ are the Gibbs densities introduced in (\ref{eq:G_k}).
\end{proposition}
The proof needs two accessory results:
\begin{lemma}\label{lemma:gentle}
For every $p\geq 1$ and $k\geq2$, we have
$$%\be\label{eq:gentle-terms}
\|G_{k,N}(\psi)\|_{L^p(\g_k)}\leq e^{\mathcal{C}} \left\|\prod_{m=0}^{k-1}
\chi_{R_m}\left(\tint h_m(\psi_N)\right)
e^{-f^k_N[\psi_N]}
\right \|_{L^p(\g_k)}\,.
$$%\ee
\end{lemma}
\begin{proof}
The lemma follows as a direct consequence from Corollary \ref{20141102:cor2},
Lemmas \ref{20150518:lem1}, \ref{20150518:lem2}, \ref{20150519:lem1}, \ref{20150519:lem2},
and Proposition \ref{MainPropSecondPart}.
\end{proof}
\begin{lemma}\label{lemma:e.net}
For $\l\geq R_0^2\sqrt{N}$ we have
$$%\be\label{eq:stima-max}
\g_k\left(\sup_{x\in\T}\left|\psi_N^{(k)}\bar\psi_N\right|\geq \l \right)\lesssim N^{2+2k}e^{-\frac{\l}{4}}\,.
$$%\ee
\end{lemma}
\begin{proof}
The proof follows from Propositions \ref{prop-Q1} and \ref{prop: sup_x} for quadratic forms in Appendix \ref{app-Gauss}. Expanding in Fourier series we see that
$$
Q_N(x):=|\psi_N^{(k)}(x)\bar\psi_N(x)|=\left|\sum_{|j|,|h|\leq N} (ih)^k e^{i(h-j)x}\psi_h\bar\psi_j\right|
$$
is a quadratic form in the Fourier coefficients of $\psi$ and it fulfills the requirement (\ref{hp:ITI}) in Proposition \ref{prop-Q1}, with $T_k\leq1$. Hence, for each $x\in\T$ we obtain 
$$%\be\label{eq:stimaP.punt}
\g_k\left(\left|\psi_N^{(k)}(x)\bar\psi_N(x)\right| \geq \l \right)\lesssim e^{-\frac{\l}{4}}\,,
$$%\ee
for all $\l>0$. Moreover, for any $x,y\in\T$, by the Cauchy--Schwarz and Bernstein inequality
\bea
|Q_N(x)-Q_N(y)|&=&\left|\int_{y}^{x}\tilde Q_N(z)dz\right|\nn\\
&\leq&\sqrt{|x-y|}\|\tilde Q_N\|_{L^2}\nn\\
&\leq&\sqrt{|x-y|}N\|Q_N\|_{L^2}\nn\\
&\leq& \sqrt{|x-y|} N^{\frac32+k} R_0^2\,.\nonumber%\label{eq:eps-stima-diff}.
\eea
Therefore we can apply Proposition \ref{prop: sup_x} with $\a=\frac12$ and $L_N=N^{\frac32+k} R_0^2$ to get for any $\e>0$ and $\l\geq N^{\frac32+k} R_0^2\sqrt{\e}$
$$%\be
\g_k\left(\sup_{x\in\T}\left|\psi_N^{(k)}(x)\bar\psi_N(x)\right|\geq \l \right)\lesssim \frac{e^{-\l/4}}{\e}.
$$%\ee
We recover the assert by setting $\e=N^{-2-2k}$.
\end{proof}

Now we can give the

\begin{proof}[Proof of Proposition \ref{Prop:densita-in-Lp}]

%{\bf 1) Gentle terms.} It is helpful to write
%$$
%\tint q_k(\psi_N)-\tint \psi_N^{(k)}\bar\psi_N^{(k-1)}\psi_N\bar\psi_N:=\tint r_k(\psi_N)
%$$
%and we have
%\be
%\left|\tint r_k(\psi_N)\right|\leq \mathcal{C}=\mathcal{C}(R_0,...,R_{k-1},k,|\b|)
%\ee
%as a direct consequence of the following Propositions and Lemmas: \ref{20140523:prop1}, \ref{20150518:lem1}, \ref{20150518:lem2}, \ref{MainPropSecondPart}, \ref{20150519:lem1} and \ref{20150519:lem2}. Therefore
%$$
%\|G_k(\psi_N)\|_{L^p(\g_k)}\leq e^{\mathcal{C}} \left\|\prod_{m=0}^{k-1}
%\chi_{R_m}\left(\tint h_m(\psi_N)\right)
%e^{-\tint \psi_N^{(k)}\bar\psi_N^{(k-1)}\psi_N\bar\psi_N}
%\right \|_{L^p(\g_k)}\,.
%$$

Let us set for brevity $\s := i \frac{2k+1}{2} \beta$.
By Lemma \ref{lemma:gentle} we have to estimate
\be\label{eq:l^p-1}
\int_0^{+\infty} t^{p-1} \g_k\left( \prod_{m=0}^{k-1} \chi_{R_m}\left(\tint h_{2m}(\psi_N)\right)
e^{- \s \tint \psi_N^{(k)}\bar\psi_N^{(k-1)}\psi_N\bar\psi_N}\geq t \right)dt\,.
\ee
We use
\bea
& &\g_k\left( \prod_{m=0}^{k-1} \chi_{R_m}\left(\tint h_{2m}(\psi_N)\right)
e^{- \s\tint \psi_N^{(k)}\bar\psi_N^{(k-1)}\psi_N\bar\psi_N}\geq t \right)\nn\\
&=&\g_k\left( \prod_{m=0}^{k-1} \chi_{R_m}\left(\tint h_{2m}(\psi_N)\right)
e^{- \s \tint \psi_N^{(k)}\bar\psi_N^{(k-1)}\psi_N\bar\psi_N}\geq t\,, |\tint h_{2m}(\psi_N)|\leq R_m\,, 0\leq m\leq k-1 \right)\nn\\
&\leq& \g_k\left( 
\left|\tint \psi_N^{(k)}\bar\psi_N^{(k-1)}\psi_N\bar\psi_N\right|\geq \frac{\ln t}{| \s |}  \,, |\tint h_{2m}(\psi_N)|\leq R_m\,, 0\leq m\leq k-1 \right)\,.\nonumber%\label{eq:start-l-picc}
\eea
It is convenient to split the integral in (\ref{eq:l^p-1}) into three parts:
\be\label{eq:3-parti}
(\ref{eq:l^p-1})=\int_0^{\exp (\s^{2} R_0^6\mathcal{C}^2)}(\cdot)+\int_{\exp (\s^{2} R_0^6\mathcal{C}^2)}^{\exp( | \s |  R_0^3\mathcal{C}\sqrt{N})}(\cdot)+\int_{\exp( | \sigma |  R_0^3\mathcal{C}\sqrt{N})}^{+\infty}(\cdot)\,.
\ee
For $t\leq e^{\s^{2} R_0^6\mathcal{C}^2}$ it suffices to use the trivial bound
\be\label{eq:trivial-Bound}
\g_k\left( 
\left|\tint \psi_N^{(k)}\bar\psi_N^{(k-1)}\psi_N\bar\psi_N\right|\geq \frac{ \ln t}{| \s |}  \,, |\tint h_{2m}(\psi_N)|\leq R_m\,, 0\leq m\leq k-1 \right)\leq 1\,.
\ee
In the range $\s^{2} R^6_0\mathcal{C}^2\leq\ln t\leq | \sigma | \mathcal{C}R^3_0\sqrt{N}$ we define 
$$
N^{*} = N^{*}(t):= \left\lfloor \frac{\ln t}{| \s | \mathcal{C} R_{0}^{3} } \right\rfloor^{2}\,, 
$$ 
noting that $N>N^*$. We decompose
\bea
& &\g_k\left( 
\left|\tint \psi_{N}^{(k)}\bar\psi_{N}^{(k-1)}\psi_{N}\bar\psi_{N}\right|\geq \frac{\ln t}{| \s |} \,, |\tint h_{2m}(\psi_N)|\leq R_m\,, 0\leq m\leq k-1\right)\nn\\
&\leq&\g_k\left( 
\left|\tint \psi_N^{(k)}\bar\psi_N^{(k-1)}\psi_N\bar\psi_N-\tint\psi_{N^{*}}^{(k)}\bar\psi_{N^{*}}^{(k-1)}
\psi_{N^{*}}\bar\psi_{N^{*}}\right|\geq \frac{\ln t}{2 |\s |} \right)\label{eq:logt<sqrtN-1}\\
&+&\g_k\left( 
\left|\tint \psi_{N^{*}}^{(k)}\bar\psi_{N^{*}}^{(k-1)}\psi_{N^{*}}\bar\psi_{N^{*}}\right|\geq \frac{\ln t}{2 | \sigma |}\right)\label{eq:logt<sqrtN-2}\,.
\eea
For the first addendum (\ref{eq:logt<sqrtN-1}), we exploit formula (\ref{eq:Cauchy-in-P}) in Lemma \ref{lemma:Cauchy-in-P}, to obtain%, WRONG  with $\sqrt{N^{*}}> (\mathcal{C}R_0^3)^{-1} \ln t$, 
\be\label{eq:lpicc-part2}
\g_k\left( 
\left|\tint \psi_N^{(k)}\bar\psi_N^{(k-1)}\psi_N\bar\psi_N-\tint\psi_{N^{*}}^{(k)}\bar\psi_{N^{*}}^{(k-1)}
\psi_{N^{*}}\bar\psi_{N^{*}}\right|\geq \frac{\ln t}{2 | \s |} \right)\lesssim t^{-\left(3 | \s | \sqrt{\mathcal{C}R_0^3}\right)^{-1}}\,.
\ee
%where $Z:=(3\sqrt{\mathcal{C}R_0^3})^{-1}$.

Since in (\ref{eq:logt<sqrtN-2}) we have $\ln t \geq | \s | R_0^3\mathcal{C}\sqrt{N^{*}}$, we can treat this term and the third addendum in (\ref{eq:3-parti}) (where we consider $\ln t\geq | \s | R_0^3\mathcal{C}\sqrt{N}$) by the same method as follows. We bound
$$%\be\label{eq:bound-quadr}
\left|\tint \psi_N^{(k)}\bar\psi_N^{(k-1)}\psi_N\bar\psi_N\right|\leq\|\psi_N^{(k)}\bar\psi_N\|_{\infty}R_0\mathcal{C}\,,
$$%\ee
whence
\bea
& & \g_k\left( 
\left|\tint \psi_N^{(k)}\bar\psi_N^{(k-1)}\psi_N\bar\psi_N\right|\geq \frac{\ln t}{| \sigma |}  \,, | \tint h_{2m}(\psi_N)|\leq R_m\,, 0\leq m\leq k-1 \right)\nn\\
&\leq& \g_k\left(\max_{x\in\T}\left|\psi_N^{(k)}\bar\psi_N\right| R_0\mathcal{C} \geq \frac{\ln t}{| \sigma |} \right)\,.
\nonumber
\eea
Thus to estimate the r.h.s. probability we use Lemma \ref{lemma:e.net} with $\l=\frac{\ln t}{| \s | R_0\mathcal{C}}$ to get
\be\label{eq:stima-logt>radN}
\g_k\left( \prod_{m=0}^{k-1} \chi_{R_m}\left(\tint h_{2m}(\psi_N)\right)
e^{- \s \tint \psi_N^{(k)}\bar\psi_N^{(k-1)}\psi_N\bar\psi_N}\geq t \right)\leq N^{2+2k}e^{-\frac{\ln t}{4 | \s | R_0\mathcal{C}}}\,.
\ee
In particular for $N=N^{*}$ we have
\be\label{eq:lpicc-part1}
\g_k\left( 
\left|\tint \psi_{N^{*}}^{(k)}\bar\psi_{N^{*}}^{(k-1)}\psi_{N^{*}}\bar\psi_{N^{*}}\right| \geq \frac{\ln t}{2 | \s |} \right)\leq (N^{*})^{2+2k}e^{-\frac{\ln t}{8 | \s | R_0\mathcal{C}}}\,.
\ee

Now we can estimate (\ref{eq:l^p-1}). We first notice that (\ref{eq:trivial-Bound}) gives
$$%\be\label{eq:p-integrabilita-bound1}
\int_0^{e^{\s^{2} R^6_0\mathcal{C}^2}} t^{p-1}\g_k\left( \prod_{m=0}^{k-1} \chi_{R_m}\left(\tint h_{2m}(\psi_N)\right)
e^{- \s \tint \psi_N^{(k)}\bar\psi_N^{(k-1)}\psi_N\bar\psi_N}\geq t \right)dt< {e^{\s^{2} R^6_0\mathcal{C}^2}}.
$$%\ee 

Then by using (\ref{eq:lpicc-part1}) and (\ref{eq:lpicc-part2}) we obtain
\bea
&&\int_{{e^{\s^{2} R^6_0\mathcal{C}^2}}}^{\exp(| \s | R_0^3\mathcal{C}\sqrt{N})} t^{p-1}\g_k\left( \prod_{m=0}^{k-1} \chi_{R_m}\left(\tint h_{2m}(\psi_N)\right)
e^{-\tint \psi_N^{(k)}\bar\psi_N^{(k-1)}\psi_N\bar\psi_N}\geq t \right)dt\nn\\
&\lesssim & \int_{{e^{\s^{2} R^6_0\mathcal{C}^2}}}^{\exp( | \s | R_0^3\mathcal{C}\sqrt{N})} t^{p-1-(8 | \s | R_0\mathcal{C})^{-1}} \ln t^{2+2k}\nn\\
&+& \int_{{e^{\s^{2} R^6_0\mathcal{C}^2}}}^{\exp( | \s |R_0^3\mathcal{C}\sqrt{N})} t^{p-1-\left(3 | \s | \sqrt{\mathcal{C}R_0^3}\right)^{-1}}\,.\nn
\eea

We note that as $p<\min\left(\left(3 | \s | \sqrt{\mathcal{C}R_0^3}\right)^{-1},(8 | \s | R_0\mathcal{C})^{-1}\right)$ both the functions on the r.h.s. are integrable, so we can bound both terms by an appropriate constant. 

Finally, using (\ref{eq:stima-logt>radN}), we have
\bea
& &\int_{\exp( | \s |R_0^3\mathcal{C}\sqrt{N})}^{+\infty} t^{p-1} 
\g_k\left( \prod_{m=0}^{k-1} \chi_{R_m}\left(\tint h_{2m}(\psi_N)\right)
e^{- \sigma \tint \psi_N^{(k)}\bar\psi_N^{(k-1)}\psi_N\bar\psi_N}\geq t \right)\nn\\
& \lesssim &
N^{2+2k} \int_{\exp( | \s | R_0^3\mathcal{C}\sqrt{N})}^{+\infty} t^{p-1} e^{-\frac{\ln t}{4 | \s |R_0\mathcal{C}}}d t\nn\\
&=& N^{2+2k} \int_{\exp( | \s |R_0^3\mathcal{C}\sqrt{N})}^{+\infty} t^{p-1-(4 | \s | R_0\mathcal{C})^{-1}}dt\nn\\
&=& \frac{N^{2+2k} e^{-[p-(4 | \s |R_0\mathcal{C})^{-1}]{\sqrt{N}}}}{|p-1-(4 | \s | R_0\mathcal{C})^{-1}|}\,,
\nonumber%\label{eq:lambda-grandi}
\eea
that vanishes for $N\to\infty$, provided that $p<(4 | \s | R_0\mathcal{C})^{-1}$.
\end{proof}
We can finally proceed to complete the proof of Theorem \ref{thm:k_even} as follows
\begin{proof}[Proof of Theorem \ref{thm:k_even}] 
The first part of the statement has been proved in Corollary \ref{20130520:cor1}.
We are left to show that $G_k(\psi)\in L^p(\g_k)$ and that
it is the $L^p(\g_k)$-limit of the sequence $G_{k,N}(\psi)$.

We start proving that $G_k(\psi)\in L^p(\g_k)$.
Let $p\geq1$ and let us choose $R_0>0$ such that Proposition
\ref{Prop:densita-in-Lp} holds.
Then there exists a subsequence $G_{k,N_m}(\psi)$, $m\in\mb Z_+$, such that
$G_{k,N_m}(\psi)\to G_{k}(\psi)$, $\g_k$-a.s. Hence, by Fatou's Lemma, we have
$$
\int |G_k(\psi)|^p\g_k(\psi)\leq \liminf_{m\to\infty}\int |G_{k,N_m}(\psi)|^p\g_k(\psi)
< \infty\,,
$$
thus proving that for $1\leq p<p_0$ given by Proposition
\ref{Prop:densita-in-Lp} $G_k(\psi)\in L^p(\g_k)$. By the uniform (for $N$ large enough) $L^p(\g_k)$-boundedness of
$G_{k,N}(\psi)$ we also have
$$
\int |G_{k,N}(\psi)-G_{k}(\psi)|^pd\g_{k}(\psi)<\infty\,.
$$
We are now ready to prove the convergence in $L^p(\g_k)$ for $p<p_0$.
For all %$N\in\mb Z_+$ and 
$\varepsilon>0$, we define
$$
A_{k,N,\varepsilon}
=\{\psi\in H^k\mid |G_{k,N}(\psi)-G_k(\psi)|\leq\varepsilon\}\,,
$$
and denote by $A^c_{k,N,\varepsilon}$ its complement.
Then let $p < q < p_0$
\begin{align*}
&\int |G_{k,N}(\psi)-G_{k}(\psi)|^pd\g_{k}(\psi)=\\
&\int_{A_{k,N,\varepsilon}}|G_{k,N}(\psi)-G_{k}(\psi)|^pd\g_{k}(\psi)
+\int_{A_{k,N,\varepsilon}^c}|G_{k,N}(\psi)-G_{k}(\psi)|^pd\g_{k}(\psi)\\
&\leq \varepsilon^p\gamma_k(A_{k,N,\varepsilon})
+\|G_{k,N}(\psi)-G_k(\psi)\|_{L^{q}(\g_k)}^p
\left(\gamma_k(A_{k,N,\varepsilon}^c)\right)^{1- p/q}\,.
\end{align*}
Since $G_{k,N}(\psi)$ converges to $G_k(\psi)$ with respect to the measure
$\gamma_k$, we have that, as $N\to\infty$,
$$
\gamma_k(A_{k,N,\varepsilon})\to 1\,,
\qquad
\gamma_k(A_{k,N,\varepsilon}^c)\to0\,,
$$
%uniformly in $\varepsilon$. 
Therefore, for a certain $\d_N$, vanishing for $N\to\infty$, we have the inequality
$$
\|G_{k,N}(\psi)-G_{k}(\psi)\|^p_{L^p(\g_k)}\leq \e^p+\d_N\|G_{k,N}(\psi)-G_k(\psi)\|^{p}_{L^{q}(\g_k)}
$$
that concludes the proof.

\end{proof}

%%%%%%%%%%%%%%%%%%%%%%%%%%%%%%%%%%%%%%%%%%
\appendix{\section{Gaussian Measures in Sobolev Spaces: a Toolbox}\label{app-Gauss}

We are here interested in giving a succinct but self contained survey on the theory of Gaussian measures in Hilbert Sobolev spaces. For a complete treatment we refer to \cite{Sko, Bog}.

%%%%%%
\subsection{Concentration of Measure in $\dot H^{k}(\T)$}

Here we study the concentration property of the Gaussian measure with covariance $(\mathbb{I}+(-\D)^{k})^{-1}$. The main feature is that the measure is concentrated on functions in $L^2(\T)$ having slightly less then $k-\frac12$ weak derivatives as regularity. This is stated precisely in the following

\begin{proposition}\label{TH:concentrato}
For every $k\geq0$ we have $\g_k\left(\bigcap_{\e>0} \dot H^{k-\frac12-\e} \right)=1$.
\end{proposition}

We will proceed by steps. At first we prove

\begin{lemma}
$\g_k(\dot H^{k-\frac12+\e})=0$ for every $\e\geq0$.
\end{lemma}
\begin{proof}
We take any function $\phi\in \dot H^{s}(\T)$ with $s\geq k-\frac12$. We have that $\|\phi_N\|_{\dot H^s}$ is finite uniformly in $N$, where we recall $\phi_N$ is the projection on the Fourier modes $|n|\leq N$ defined by (\ref{eq:proiettore}) and (\ref{eq:u_N}). We show that for all $\l>0$
$$%\be
\g_k\left(\|\phi_N\|_{\dot H^s}\leq \l \right)\to0,\quad\mbox{as}\quad N\to\infty.
$$%\ee
To do so, we make use of the Markov inequality: for every $\mu>0$
\bea
\g_k\left(\|\phi_N\|_{\dot H^s}\leq \l \right)&\leq& e^{\frac{\mu\l}{2}} \int\prod_{|n|\leq N}\left(\frac{1+n^{2k}}{\sqrt{2\pi}}d\phi_n d\bar\phi_n\right)e^{-\frac12\sum_n(1+n^{2k})|\phi_n|^2}e^{-\frac\mu2\sum_n n^{2s}|\phi_n|^2}\nn\\
&\leq&e^{\frac{\mu\l}{2}} \int\frac{1}{(2\pi)^{2N+1}}d\phi'_n d\bar\phi'_n e^{-\frac12\sum_n|\phi'_n|^2}e^{-\mu\sum_n n^{2(s-k)}|\phi'_n|^2}\nn\\
&=&\exp\left[\frac{\mu\l}{2}-\frac12\sum_{\substack{|n|\leq N, \\ n\neq 0}}\ln\left(1+\frac{\mu}{|n|^{\kappa}}\right)\right]\,,
\nonumber
\eea
where we have performed the change of variables $\phi'_n=\sqrt{1+n^{2k}}\phi_n$ and set $-2(k-s)=:\kappa$. Let us first consider negative $\kappa$. In this case 
$$
\sum_{\substack{|n|\leq N, \\ n\neq 0}}\ln\left(1+\frac{\mu}{|n|^{\kappa}}\right)\geq 2N\ln(1+\mu),
$$
and so we have an exponential decay in $N$ for every choice of positive $\mu$:
\be\label{eq:gamma-k<0}
\g_p\left(\|\phi_N\|_{\dot H^{s}}\leq \l \right)\lesssim e^{\frac{\mu\l}{2}} e^{-2N \ln(1+\mu)},\qquad(s>k).
\ee
For $\kappa\in[0,1)$ the series $\sum_n\ln\left(1+\frac{\mu}{n^{\kappa}}\right)$ diverges as $N^{1-\kappa}$. Hence
\be\label{eq:gamma-k<1}
\g_k\left(\|\phi_N\|_{\dot H^{s}}\leq \l \right)\lesssim e^{-\mu N^{1-2(k-s)}},\qquad\left(k\geq s>k-\frac12\right).
\ee
Finally for $\kappa=1$ we have a logarithmic divergence at exponent and therefore
\be\label{eq:gamma-k=1}
\g_k\left(\|\phi_N\|_{\dot H^{s}}\leq \l \right)\leq \left(\frac{e^{\frac{\l}{2}}}{N}\right)^{\mu},\qquad\left(s=k-\frac12\right),
\ee
for arbitrary $\mu>0$. We obtain the statement by taking $N\to\infty$ in (\ref{eq:gamma-k<0}), (\ref{eq:gamma-k<1}) and (\ref{eq:gamma-k=1}).
\end{proof}

\begin{remark}
The same strategy can be also used to show the stronger statement
$$%\be
\g_k\left(\|\phi_N\|_{\dot H^s}\leq \ln N \right)\to0\mbox{ as }N\to\infty,\qquad\left(s\geq k-\frac12\right).
$$%\ee
\end{remark}
\begin{lemma}\label{Lemma:s<k-1/2}
We have that for every $s<k-\frac12$ and $\l>0$
\be\label{eq:conc-12-e}
\g_k\left(\|\phi\|_{\dot H^s}\geq \l \right)\lesssim e^{-\l/4}.
\ee
\end{lemma}
\begin{proof}
Let us take a function $\phi\in \dot H^{s}$ for some $s<k-\frac12$. We look at its truncation $\phi_N$ and again it is $\|\phi_N\|_{\dot H^s}$ finite uniformly in $N$. We exploit the reverse Chernoff bound at finite $N$: for every $\mu\in(0,1)$ and $\l>0$, we get
\bea
\g_k\left(\|\phi_N\|_{\dot  H^s}\geq \l \right)&\leq& e^{-\frac{\mu\l}{2}} \int\prod_{|n|\leq N}\left(\frac{1+n^{2k}}{\sqrt{2\pi}}d\phi_n d\bar\phi_n\right) e^{-\frac12\sum_n(1+n^{2k})|\phi_n|^2}e^{\frac\mu2\sum_n n^{2s}|\phi_n|^2}\nn\\
&\leq&e^{-\frac{\mu\l}{2}} \int\frac{1}{(2\pi)^{2N+1}}d\phi'_n d\bar\phi'_n e^{-\frac12\sum_n|\phi'_n|^2}e^{\mu\sum_n n^{2(s-k)}|\phi�_n|^2}\nn\\
&=&\exp\left[-\frac{\mu\l}{2}-\frac12\sum_{\substack{|n|\leq N, \\ n\neq 0}}\ln\left(1-\frac{\mu}{|n|^{\kappa}}\right)\right]\,,
\nonumber
\eea
where again we have used the same change of variables as before. Note that now it is $\kappa>1$. Since $\frac12\sum_n\ln(1-\frac{\mu}{n^{\kappa}})$ is convergent for all $\mu<1$ and $\kappa>1$, we can choose $\mu\in(0,1)$ and take the limit $N\to\infty$. We get (\ref{eq:conc-12-e}) by setting $\mu=1/2$.
\end{proof}

Equation (\ref{eq:conc-12-e}) implies that $\|u\|_{H^{q}}$ is bounded with probability 1 for every 
$k<s-\frac12$. This is sufficient to complete the proof of Proposition \ref{TH:concentrato}.

%%%%%%%%%%
\subsection{Quadratic  Forms}

Then we present some results about quadratic forms of Gaussian random variables, used in the paper.

%\begin{proposition}\label{prop:linear}
%Let $L$ be a linear form on $\R^N$:
%$$
%Lu:=\sum_{n=1}^N c_i u_i\,,
%$$
%with $c_i\in\mathbb{C}$, $i=1,\dots,N$. We put for each $k\geq1$
%$$
%b_N:=\sum_{n=1}^N \frac{|c_n|^2}{(1+n^{2k})}\,.
%$$
%Then it is
%\be\label{eq:linear}
%\g_k\left(|Lu|\geq\l\right)\leq 2e^{-\frac{\l^2}{2b_N}}\,.
%\ee
%\end{proposition}
%\begin{remark}
%In particular if $Lu=u^{(s)}$ we have $b_N<\infty$ uniformly in $N$ for $s<k-\frac12$.
%\end{remark}
%\begin{proof}
%Let us assume $Lu\geq0$. Then we have $\forall \mu>0$
%\bea
%\g_k\left(|Lu|\geq\l\right)&\leq& e^{-\mu\l}\E{e^{\mu Lu}}\nn\\
%&=&e^{-\mu\l}\int\prod_{|n|\leq N}\left(\frac{1+n^{2k}}{\sqrt{2\pi}}d\phi_n d\bar\phi_n\right)\exp\left[-\frac12\sum_{i}|\phi_i|^2(1+j^{2k})+\mu c_i\phi_i\right]\nn\\
%&=&e^{-\mu\l+\frac12\mu^2 b_N}
%\eea
%by means of a change of variables 
%$\phi'_j=\phi_j/\sqrt{1+j^{2k}}, \bar{\phi}'_j=\bar{\phi}_j/\sqrt{1+j^{2k}}$ 
%followed by a Gaussian integration. By optimising over $\mu$ we get a decay as in (\ref{eq:linear}). The case $Lu\leq0$ is dealt analogously.
%\end{proof}

\begin{proposition}\label{prop-Q1}
Let $k \geq 2$ and $Q$ be a $(2N+1)\times (2N+1)$ matrix such that 
\begin{equation}\label{hp:ITI}
 \sup_{l,h } \frac{|Q_{l h}|}{\sqrt{1+h^{2k}}} =: T_k < + \infty
\end{equation}
Then for $\l>0$
\be\label{eq:dec-Quadr}
\g_k\left((\phi,Q\phi)\geq\l\right)\lesssim 
e^{-\l/4T_k}\,.
\ee
\end{proposition}
\begin{proof}
To begin with, we exploit the Markov inequality: for any $\mu>0$ 
\begin{equation}\label{FTP}
\g_k\left((\phi,Q\phi)\geq\l\right)\leq e^{-\mu\l}\E{e^{\mu (\phi,Q\phi)}}.
\end{equation}
Now we compute
\begin{eqnarray}\nonumber\label{FT}
\E{e^{\mu (\phi, Q \phi)}}&=&\int\prod_{|n|\leq N}\left(\frac{1+n^{2k}}{\sqrt{2\pi}}d\phi_n d\bar\phi_n\right)\exp\left[-\frac12\sum_{i,j}\bar\phi_i\left((1+j^{2k})\d_{ij}-2\mu Q_{ij}\right)\phi_j\right]
\\ \nonumber
&=& 
\frac{1}{(2\pi)^{N}}
\int d\phi'_{-N}...d\phi'_N\bar\phi'_{-N}...d\bar\phi'_N\exp\left[-\frac12\sum_{i,j}\bar\phi'_i\left(\d_{ij}-2\mu Q_{ij}(k)\right)\phi'_j\right]
\\
&=& e^{-\frac 12\ln\det(\mathbb{I}-2\mu Q(k))}\label{eq:det(1-Q)} .
\end{eqnarray}
where we have performed the change of variables $\phi'_j=\sqrt{1+j^{2k}}\phi_j$,
$\bar{\phi}'_j=\sqrt{1+j^{2k}}\bar{\phi}_j$ and we have introduced $ Q_{ij}(k):=Q_{ij}/\sqrt{(1+j^{2k})(1+i^{2k})}$. 
We claim that 
\begin{equation}\label{claim}
| \Tr ((Q_{ij}(k))^{m}) |  \lesssim T_{k}^{m}\,,
\quad m \in \mathbb{Z}_{+}\,, 
\end{equation}
so the expansion of the determinant
$$%\begin{equation}
-\ln\det(\mathbb{I} - 2 \mu  Q(k))
= \sum_{m=1}^{+\infty}  \frac{(2 \mu )^{m} \Tr (( Q_{ij}(k))^{m})}{m}
$$%\end{equation}
is convergent provided that $\mu < \frac{1}{2 T_{k}}$. 
We choose $\mu = \frac{1}{4T_{k}}$, so that (\ref{FTP}, \ref{FT}) imply the desired inequality.  
It remains to show the (\ref{claim}).
\begin{eqnarray}\nonumber
\Tr( (Q(k) )^{m}) & = & 
\sum_{i_{1}, \dots, i_{m+1}} Q_{i_{1} i_{2}}(k)  \dots  Q_{i_{m} i_{m+1}}(k) \delta_{i_{1} i_{m+1}}
\\ \nonumber
& = &
\sum_{i_{1}, \dots, i_{m+1}} 
\frac{ 
Q_{i_{1} i_{2}}  \dots  Q_{i_{m} i_{m+1}}   \delta_{i_{1} i_{m+1}}  
} 
{
\sqrt{ (1+i_{1}^{2k}) (1+i_{2}^{2k}) (1+i_{2}^{2k})\dots (1+i_{m}^{2k})(1+i_{m}^{2k}) (1+i_{m+1}^{2k}) }
}
\\ \nonumber
& \leq &
T_{k}^{m}\sum_{i_{1},\dots , i_{m+1}} 
\frac{ 
\delta_{i_{1}, i_{m+1}}  
} 
{
\sqrt{ (1+i_{1}^{2k}) \dots (1+i_{m}^{2k}) } 
}
\\  \nonumber
& = &
T_{k}^{m}\sum_{i_{1}, \dots, i_{m}} 
\frac{ 
1   
} 
{
\sqrt{ (1+i_{1}^{2k}) \dots (1+i_{m}^{2k})}
}
\\  \nonumber
& = &
T_{k}^{m}\left(\sum_{i} 
 \frac{ 
1   
} 
{
\sqrt{ (1+i^{2k})  }    
}
\right)^{m}  \lesssim T_{k}^{m} ,
\end{eqnarray} 
where we have used the assumption (\ref{hp:ITI}) in the first inequality and and $k \geq 2$ in the last inequality. 
\end{proof}
\begin{remark}
We observe that we can make different assumptions on the matrix $Q$ and obtain similar inequalities. 
For instance, if the trace norm of $ Q(k)$ is finite uniformly in $N$, we have (see for instance Lemma 3.3 in \cite{trace})
$$
-\ln\det(\mathbb{I} - 2 \mu  Q(k))\leq\|  Q(k)\|_{\Tr}
$$
and so for every $N$
\be
\g_k\left((\phi, Q \phi)\geq\l\right)\lesssim e^{-\lambda  /  \| Q(k) \|_{\Tr}}\,,
\ee
by the same argument of the last proposition. If we assume the Hilbert-Schmidt norm of $Q(k)$ to be finite uniformly in $N$, we obtain the Hanson-Wright inequality (see \cite{HW} and more recently \cite{RV}), holding for any $N$
\be
\g_k\left(\Var(\phi,Q\phi)\geq\l^2\right)\lesssim e^{-c\min\left(  \lambda^2/ \| Q(k)\|_{HS}, \lambda /  \| Q(k)\|\right)}\,,
\ee
where $\|A\|$ denotes the operator norm of $A$ and $c$ is a positive constant. 
\end{remark}

\begin{remark}
For any linear operator $A\phi:=\sum_{n=1}^N a_i \phi_i$, with 
$$
\mathcal{T}_k:=\sum_{|i|\leq N} \frac{|a_i|^2}{(1+i^{2k})}<\infty\quad \mbox{uniformly in $N$}\,,
$$
by using $\g_k(|A\phi|\geq\l)=\g_k(|A\phi|^2\geq\l^2)$ we can infer
\be\label{eq:linear}
\g_k\left(|A\phi|\geq\l\right)\lesssim e^{-\l^2/\mathcal{T}_k}\,.
\ee
Note that if $A\phi=\phi^{(s)}$ we have $\mathcal{T}_k<\infty$ uniformly in $N$ for $s<k-\frac12$. In this way we can improve Lemma \ref{Lemma:s<k-1/2}, obtaining a sub-Gaussian decay.
\end{remark}

\begin{proposition}\label{prop: sup_x}
Let $Q(x)$ be a $N\times N$ matrix as before. Moreover we assume $Q(x)$ to be H\"older continuous w.r.t. $x\in\T$ with exponent $\a$ and constant $L_N$, \ie
\be\label{eq:cond-Holder}
|(\phi,Q(x)\phi)-(\phi,Q(y)\phi)|\leq L_N|x-y|^\a, \quad\mbox{for every }\phi\in\R^N.
\ee
Then for any $\e>0$ and $\l\geq 2L_N \e^\a$
\be\label{eq:sup-Q}
\g_k\left(\sup_{x}(\phi,Q\phi)\geq\l\right)\lesssim \frac{e^{-\l/4T_k}}{\e}.
\ee
\end{proposition}
\begin{proof}
We exploit Proposition \ref{prop-Q1} along with an $\e$-net argument. For $\e>0$ we divide the interval $\T$ in $1/\e$ points at distance $\e$. We denote by $x_j$ a point in the $j$-th segment, and by $x^*$ the point in which the maximum is attained. By Proposition \ref{prop-Q1} for each $x\in\T$ we obtain for $\l>0$
\be\label{eq:Q-puntuale}
\g_k\left((\phi,Q(x)\phi)\geq\l\right)\lesssim e^{-\l/4T_k}\,.
\ee
Let $j_0$ be such that $|x_{j_0}-x^*|\leq \e$. Therefore it has to be
$$
|(\phi,Q(x^*)\phi)-(\phi,Q(x_{j_0})\phi)|\leq L_N\e^\a, \quad\mbox{for every }\phi\in\R^N.
$$

Then we use the union bound for the probabilities:
\bea
\g_k\left(|Q(x^*)|\geq \l \right)&\leq& \sum_j \g_k\left(|Q(x^*)|\geq \l \: \Big|\: |x^*-x_j|\leq\e\right)\nonumber%\label{eq:union.bound}
\\
&\leq& \sum_j \g_k\left(|Q(x_j)|\geq \frac{\l}{2}\right)\nn\\
&+&\sum_j \g_k\left(|Q(x_j)-Q(x^*)|\geq \frac{\l}{2} \: \Big|\: |x^*-x_j|\leq\e\right)\nn.
\eea
We immediately see by \eqref{eq:cond-Holder} that the second addendum in the last formula is zero as soon as $\l\geq 2L_N \e^\a$. Therefore we bound the first addendum by the total number of terms in the sum, which is $\e^{-1}$, times the estimate (\ref{eq:Q-puntuale}), so obtaining (\ref{eq:sup-Q}).
\end{proof}

}

%%%%%%%%%%% BIBLIOGRAFIA %%%%%%%%%%%%%%%%%%%%%

%
\end{document}